\documentclass[a4paper, 11pt]{article} \pdfoutput=1

\usepackage{amsmath, amsfonts, amscd, amssymb,  amsthm, mathrsfs, ascmac}
\usepackage{latexsym}
\usepackage{longtable, geometry}
\usepackage[english]{babel}
\usepackage[utf8]{inputenc}
\usepackage{graphicx}
\usepackage{bbm}
\usepackage{enumitem}
\usepackage{stmaryrd}
\usepackage{nicefrac}
\usepackage{bm}
\usepackage{physics}
\usepackage{mathtools} 
\usepackage{empheq}

\usepackage[subrefformat=parens]{subcaption}
\usepackage[svgnames,psnames]{xcolor}
\usepackage[colorlinks,citecolor=DarkGreen,linkcolor=FireBrick,linktocpage,unicode]{hyperref}
 \usepackage{braket}
\usepackage{xparse}
\usepackage{tikz-cd}
\usetikzlibrary{intersections,calc,arrows}
\usetikzlibrary{decorations.pathreplacing}
\usetikzlibrary{positioning}

\usepackage{url}
\hypersetup{urlcolor=blue}
\usepackage{comment}
\usepackage{authblk}
\newtheorem{Thm}{Theorem}[section]

\newtheorem{Lemm}[Thm]{Lemma}
\newtheorem{Prop}[Thm]{Proposition}
\newtheorem{Coro}[Thm]{Corollary}
\theoremstyle{definition}

\newtheorem{Rem}[Thm]{Remark}
\newcommand{\be}{\begin{equation}}
\newcommand{\ee}{\end{equation}}

\newcommand{\ba}{\begin{align}}
\newcommand{\ea}{\end{align}}

\newcommand{\ben}{\begin{equation*}}
\newcommand{\een}{\end{equation*}}

\def\i<#1>{\langle #1 \rangle}
\def\l<#1>{\left\langle #1 \right\rangle}
\def\b<#1>{\big\langle #1 \big\rangle}
\def\B<#1>{\Big\langle #1 \Big\rangle}

\newcommand{\deq}{\stackrel{\mathrm{def}}{=}}

\newcommand{\im}{\mathrm{i} }
\newcommand{\la}{\langle}
\newcommand{\ra}{\rangle}
\newcommand{\bs}{\boldsymbol}
\newcommand{\ex}{\mathrm{e}}
\newcommand{\h}{\mathfrak{H}}

\newcommand{\D}{\mathrm{dom}}
\newcommand{\BbbR}{\mathbb{R}}

\newcommand{\BbbZ}{\mathbb{Z}}

\newcommand{\BbbC}{\mathbb{C}}

\newcommand{\vepsilon}{\varepsilon}
\newcommand{\vphi}{\varphi}
\newcommand{\no}{\nonumber \\}
\newcommand{\Ex}{\mathbb{E}}

\newcommand{\vOm}{\varOmega}
\newcommand{\vLa}{\varLambda}

\newcommand{\vPhi}{\varPhi}

\newcommand{\He}{H_{\vepsilon}}

\newcommand{\vPsi}{\varPsi}

\newcommand{\F}{\mathfrak{F}}
\newcommand{\E}{\mathfrak{E}}
\newcommand{\K}{\mathfrak{K}}

\newcommand{\one}{{\mathchoice {\rm 1\mskip-4mu l} {\rm 1\mskip-4mu l}
{\rm 1\mskip-4.5mu l} {\rm 1\mskip-5mu l}}}

\makeatletter
  
  \@addtoreset{equation}{section}
\makeatother

\title{ \bf
Ordering of Energy Levels in the Fr\"{o}hlich Model 
}

\date{}
\author[b]{Fumio Hiroshima}
\author[a]{Akihiro Kobayashi}

\author[a]{Tadahiro Miyao}
\author[a]{Shunsuke Tomioka}

\affil[a]{Department of Mathematics,  Hokkaido University}

\affil[b]{Faculty of Mathematics, Kyushu University}

\begin{document}
\maketitle

\begin{abstract}
Consider a one-dimensional system of \( N \) electrons subject to an external potential \( U \). Let \( E_{\rm el}(S) \) denote the ground state energy of the system with total spin \( S \). The Mattis--Lieb theorem asserts that, for a broad class of potentials \( U \), the inequality \( E_{\rm el}(S) < E_{\rm el}(S') \) holds whenever \( S < S' \). This result implies that the ground state of a one-dimensional many-electron system is non-ferromagnetic.
In the present work, we demonstrate that the Mattis--Lieb theorem can be extended to electron-phonon interacting systems governed by the Fröhlich model. Our analysis is carried out in the setting without an ultraviolet cutoff.   The cornerstone of our approach  is the construction of a Feynman--Kac-type formula for the heat semigroup generated by the Fröhlich Hamiltonian.
\end{abstract}


\section{Introduction}

 The quantum origins of metallic ferromagnetism are currently the subject of numerous theoretical explanations, often leading to the misconception that all issues have been thoroughly resolved. However, a particularly rigorous mathematical understanding remains insufficient, and mathematical physicists are actively engaged in addressing this problem.
 \medskip

 A central challenge in comprehending metallic ferromagnetism lies in the mathematical description of phenomena arising from the collaboration of numerous electrons possessing four fundamental properties: spin, Fermi statistics, itinerancy, and Coulomb interactions. Integrating these elements to elucidate metallic ferromagnetism remains an exceedingly complex issue.\footnote{For a rigorous analysis of metallic ferromagnetism, the textbook \cite{Tasaki2020} provides an excellent introduction. For a theoretical physics perspective on metallic ferromagnetism, see \cite{Mattis2006}.}
 \medskip
 
 In the early studies of metallic ferromagnetism theory, Mattis and Lieb made significant contributions by proving that, for one-dimensional \(N\)-electron systems governed by the Schrödinger operator, the ground state is non-magnetic regardless of the specific form of the inter-electronic potential \cite{Lieb1962}. In Section \ref{Sec2}, this paper elaborates on the results of Mattis and Lieb. Furthermore, their theorem applies to the Hubbard model and the  Heisenberg model  on a linear chain \cite{Lieb1962, Lieb1962-2}. These findings indicate a strong tendency for one-dimensional many-electron systems to avoid ferromagnetism under physically natural conditions.
 \medskip

 Electrons in metals are subject to various environmental perturbations, such as thermal fluctuations and lattice vibrations (phonons). These perturbations might  imply the possibility of ferromagnetism even in one-dimensional systems. However, Aizenman and Lieb extended the Mattis--Lieb theorem to finite temperatures \cite{AizenmanLieb}. Moreover, one of the authors has demonstrated that the Mattis--Lieb theorem for the lattice Hubbard model can be further extended to models incorporating electron-boson interactions, such as the Holstein--Hubbard model \cite{Miyao2012-2}. These results further reinforce the  conclusion that the emergence of ferromagnetic order in one-dimensional many-electron systems is highly unlikely.
 \medskip

The analysis of electron-phonon interaction systems discussed above pertains to the lattice Hubbard model, and there has been minimal mathematical research concerning the magnetism of continuous one-dimensional many-electron systems described by Schrödinger operators, particularly when accounting for their interactions with phonons. The objective of this paper is to demonstrate that the Mattis--Lieb theorem holds for the Fröhlich model, a fundamental framework for describing electron-phonon interactions in continuous systems \cite{Frhlich1954}. This proof is expected to further deepen our understanding of the non-magnetic tendencies exhibited by one-dimensional many-electron systems. Moreover, we will provide additional analyses regarding the ground state energy of the Fröhlich model, specifically establishing that electron-phonon interactions indeed lower the energy of many-electron systems.
 \medskip

 A technical challenge in the analysis of the Fröhlich model arises from weak ultraviolet (UV) singularities. UV singularities refer to the necessity of energy renormalization when removing ultraviolet cutoffs from the Hamiltonian. The Nelson model \cite{Nelson1964} serves as a typical example, and significant advancements in its analysis have been made recently, see, e.g.,  \cite{Gubinelli2014, Hasler2024, Hiroshima2021, Lampart2021, MATTE2018, Miyao2019-2}. In the context of the Fröhlich model, it is possible to define the Hamiltonian without an ultraviolet cutoff by employing Nelson's techniques; however, it has been demonstrated that energy renormalization is unnecessary for this model. In this sense, the Fröhlich model exhibits \lq\lq{}weak" UV singularities.
 \medskip
 
In order to overcome the aforementioned weak UV singularities, this paper applies the functional integral techniques developed by Gubinelli, L\"{o}rinczi, and Hiroshima \cite{Gubinelli2014}, as well as by Matte and M{\o}ller \cite{MATTE2018}, originally designed for the analysis of ultraviolet cutoff removal in the Nelson model, to examine the ground state energy of the Fröhlich model.
\medskip

Another significant challenge in the analysis of many-electron systems, as addressed herein, is overcoming the sign problem arising from the intrinsic fermionic statistics of electrons. This issue can be mitigated by leveraging the unique characteristics of one-dimensional systems. Numerous studies have demonstrated that one-dimensional fermionic systems exhibit bosonic properties, with bosonization techniques serving as a prominent example \cite{Langmann2015, Mattis1965}. In this work, we extend the bosonic tendencies discovered by Mattis and Lieb in one-dimensional systems \cite{Lieb1962, Lieb1962-2} and elucidate an effective approach utilizing path integral representations.
A direct and illustrative consequence of this approach is the uniqueness of the ground state---an outcome that cannot generally be expected for fermionic systems in dimensions two or higher; see Theorem \ref{ExGs}.
\medskip

The Fröhlich model has been well-known for an extended period, characterized by its simplicity in definition. It exhibits numerous intriguing properties, not only from a physical standpoint but also from a mathematical perspective. One such property is the previously mentioned weak UV singularity. As a result, the mathematical aspects of this model have continued to capture the interest of researchers in mathematical physics, see, e.g., \cite{Donsker1983, Lieb1997, Spohn1987}. 
In recent years, it has continued to present a wealth of problems for the mathematical physics community, such as the analysis of the binding energy of many-body polarons \cite{FLST, Miyao2007} and the effective mass of polarons in the strong coupling regime \cite{Betz2022, Brooks2024,Dybalski2020,Mukherjee2020}, fostering vibrant research activity. The mathematical investigation of the magnetism of the Fröhlich model presented in this study is unprecedented, and we anticipate that it will reveal new facets of this model.
\medskip

The structure of this paper is as follows.  
In Section \ref{Sec2}, we first introduce the Mattis--Lieb theorem. Subsequently, we present the Fröhlich model and provide a detailed exposition of the main theorems: the existence and uniqueness of the ground state in this system (Theorem \ref{ExGs}), the extension of the Mattis--Lieb theorem to this model (Theorem \ref{EnOrd}), and the energy-lowering effect of electron-phonon interactions (Theorem \ref{GEBasic}).  
Section \ref{SecMsubspace} explores in detail the Hilbert space of states for one-dimensional many-electron systems, with particular focus on the \( M \)-subspace. The results obtained here play a crucial role in the subsequent sections.  
In Section \ref{PfofThmEnOrd}, we present a proof of Theorem \ref{EnOrd} within a somewhat abstract framework.  
Section \ref{FKIF} constructs the path integral representation of the heat semigroup generated by the Fröhlich model. This representation proves to be useful in the subsequent sections.  
In Section \ref{PfExGs}, we provide a proof of Theorem \ref{ExGs}.  
In Section \ref{PfofThmGEBasic}, we further utilize the constructed path integral representation to establish Theorem \ref{GEBasic}.  
Finally, Section \ref{ConRem} contains concluding remarks.
The appendices present several key results that are essential for the developments in the main text.

\subsection*{Acknowledgements}
F. H.  was financially supported by JSPS KAKENHI Grant Numbers JP20H01808, JP20K20886 and JP25H00595.
T. M.  was  supported by JSPS KAKENHI Grant Numbers  20KK0304 and 	23H01086.  T. M. expresses deep gratitude for the generous hospitality provided by Stefan Teufel at the Department of Mathematics, University of T\"{u}bingen, where T. M. authored a segment of this manuscript during his stay.

\subsection*{Declarations}
\begin{itemize}
\item  Conflict of interest: The Authors have no conflicts of interest to declare
that are relevant to the content of this article.
\item  Data availability: Data sharing is not applicable to this article as no
datasets were generated or analysed during the current study
\end{itemize}

\section{Main Results}\label{Sec2}

\subsection{Previous Studies}
To provide a clearer understanding of the main results presented in this paper, we first review the Mattis--Lieb theorem \cite{Lieb1962}.
Let $L > 0$, and consider a one-dimensional $N$-electron system on the interval $D \deq (-L, L)$. This system is described by the Schrödinger operator:  
\be  
H_{\rm el} \deq \sum_{j=1}^N \left\{  
-\frac{1}{2} \Delta_{\mathrm{D}, j} + V(x_j)  
\right\} + \sum_{1 \le i < j \le N} W(x_i - x_j). \label{DefHel}  
\ee  
Here, $\Delta_{\mathrm{D}, j}$ denotes the Dirichlet Laplacian on $D$.  
The operator $H_{\rm el}$ acts on the Hilbert space:
\be  
\mathfrak{E} \deq \bigwedge^N L^2(\vLa), \quad \vLa\deq  D \times \{-1, +1\}, \notag
\ee  
where, $-1$ (resp. $+1$) corresponds to an electron with spin down (resp. spin up).
 Additionally,  
\be  
L^2(\vLa) \deq \left\{  
f : \vLa \to \BbbC, \sum_{\sigma \in \{-1, +1\}} \int_D dx \, |f(x, \sigma)|^2 < \infty  
\right\}.  \notag
\ee  
For a given Hilbert space $\mathfrak{X}$, $\bigwedge^N \mathfrak{X}$ denotes the $N$-fold anti-symmetric tensor product of $\mathfrak{X}$. A detailed explanation of the action of \( H_{\rm el} \) on \( \mathfrak{E} \) is given in Remark \ref{Rem1}, as some preliminary considerations are required for a full account.

The potentials \( V \) and \( W \) satisfy the following assumptions:

\begin{description}
\item[\hypertarget{H1}{\bf (H.1)}]  
The potential \( W \) is symmetric in the sense that \( W(x) = W(-x) \) for almost every \( x \in D \).  
Define the function \( U \colon D^N \to \mathbb{R} \) by  
\[
U(x_1, \dots, x_N) \deq\sum_{j=1}^N V(x_j) + \sum_{1 \le i < j \le N} W(x_i - x_j).  
\]

The kinetic energy operator is given by  
\[
T_{D^N} \deq \sum_{j=1}^N \left(-\frac{1}{2} \Delta_{\mathrm{D}, j}\right).  
\]

The multiplication operator \( U \) is infinitesimally small with respect to \( T_{D^N} \).  
\end{description}

Since $T_{D^N}$ has a compact resolvent, $H_{\rm el}$ also possesses a compact resolvent. Consequently, $H_{\rm el}$ has purely  discrete spectrum.

To simplify the subsequent expressions,  we introduce a useful representation for the elements of $\mathfrak{E}$. The following identifications are commonly used:  
\begin{align}  
\bigotimes^N L^2(\vLa) = \bigotimes^N \left( L^2(D) \otimes \mathbb{C}^2 \right) &= L^2\left(D^N\right) \otimes \left( \bigotimes^N \mathbb{C}^2 \right), \label{IdnHi1} \\  
(f_1 \otimes \eta_{\sigma_1}) \otimes \cdots \otimes (f_N \otimes \eta_{\sigma_N}) &= \left( f_1 \otimes \cdots \otimes f_N \right) \otimes \eta_{\bs{\sigma}}.  \notag
\end{align}  
Here, for $\bs{\sigma} = (\sigma_1, \dots, \sigma_N) \in \{-1, +1\}^N$, we define  
$
\eta_{\bs{\sigma}} \deq \eta_{\sigma_1} \otimes \cdots \otimes \eta_{\sigma_N}, 
$
where $\eta_{-1}$ and $\eta_{+1}$ represent spin-down and spin-up states, respectively:  
\be
\eta_{-1} \deq  
\begin{bmatrix}  
0 \\ 1  
\end{bmatrix}, \quad  
\eta_{+1} \deq   
\begin{bmatrix}  
1 \\ 0  
\end{bmatrix}. \notag
\ee
Under this identification, the following identity holds in \(\mathfrak{E}\):
\[
(f_1 \otimes \eta_{\sigma_1}) \wedge \cdots \wedge (f_N \otimes \eta_{\sigma_N})
= \sum_{\pi \in \mathfrak{S}_N} \frac{\mathrm{sgn}(\pi)}{N!}\,
S_{\pi}\bigl(f_1 \otimes \cdots \otimes f_N\bigr)
\otimes \bigl(s_{\pi}\,\eta_{\bs{\sigma}}\bigr),
\]
where the unitary operators \(S_{\pi}\) on \(\otimes^N L^2(D)\) and \(s_{\pi}\) on \(\otimes^N \mathbb{C}^2\) represent the permutation action; That is, let \(\mathfrak{S}_N\) denote the symmetric group on \(N\) elements. For each \(\pi \in \mathfrak{S}_N\), define unitary operators by
\[
S_{\pi}\bigl(f_1 \otimes \cdots \otimes f_N\bigr)
\deq f_{\pi(1)} \otimes \cdots \otimes f_{\pi(N)},
\quad
f_1, \dots, f_N \in L^2(D),
\]
and
\[
s_{\pi}\,\eta_{\bs{\sigma}}
\deq  \eta_{\pi \bs{\sigma}},
\quad
\pi \bs{\sigma} \deq  (\sigma_{\pi(1)}, \dots, \sigma_{\pi(N)}).
\]

\begin{Rem}\label{Rem1}
\upshape
A more precise definition of $H_{\rm el}$ is as follows:  
Let $h$ denote the self-adjoint operator on $L^2(D^N)$ defined by the right-hand side of Eq.  \eqref{DefHel}.  
Then, using the identification  Eq. \eqref{IdnHi1}, $H_{\rm el}$ is defined by  
\[
H_{\rm el} = \left( h \otimes \mathbbm{1}_{\otimes^N \BbbC} \right) \restriction \mathfrak{E}.
\]  
Henceforth, in this paper, we shall abbreviate expressions like $A \otimes \mathbbm{1}$ and $\mathbbm{1} \otimes B$ to $A$ and $B$, respectively, whenever the context precludes any ambiguity.
\end{Rem}

The total spin operators $S^{(1)}, S^{(2)}, S^{(3)}$, acting on $\otimes^N \mathbb{C}^2$, are defined as  
\[ 
S^{(i)} \deq  \frac{1}{2} \sum_{j=1}^N \one \otimes \cdots \otimes \one \otimes \overbrace{\sigma^{(i)}}^{j^{\text{th}}} \otimes \one \otimes \cdots \otimes \one, \quad (i=1, 2, 3),
\]  
where $\sigma^{(1)}, \sigma^{(2)}, \sigma^{(3)}$ denote the standard Pauli matrices.  
In particular, the following relation holds:  
\[ 
\sigma^{(3)} \eta_{\sigma} = \sigma \eta_{\sigma}. 
\]  
By denoting $\one \otimes S^{(i)}$ simply as $S^{(i)}$, these operators are naturally extended to act on $\bigwedge^N L^2(\vLa)$. The Casimir operator is defined by  
\[ 
\bs{S}^2 \deq  (S^{(1)})^2 + (S^{(2)})^2 + (S^{(3)})^2. 
\]  
A state vector $ \vPsi $ is said to have {\bf total spin} $ S $ if  
\[
\bs{S}^2 \vPsi = S(S+1) \vPsi.
\]  
Let $ \mathfrak{E}^S $ denote the subspace of $ \mathfrak{E} $ consisting of state vectors with total spin $ S $. For each $ N $, the set of possible values of total spin is denoted by $ \mathbb{S}_N $, and is given by:
\[
\mathbb{S}_N \deq
\begin{cases}
\{N/2, N/2-1, \dots, 0\} & \text{if } N \text{ is even}, \\
\{N/2, N/2-1, \dots, 1/2\} & \text{if } N \text{ is odd}.
\end{cases}
\]

\begin{Thm}[Mattis--Lieb \cite{Lieb1962}]  
For each $ S \in \mathbb{S}_N $, define  
$
E_{\rm el}(S) \deq  \inf \mathrm{spec}(H_{\rm el} \restriction \mathfrak{E}^S).
$  
Then, the following holds:  
\[
E_{\rm el}(N/2) > E_{\rm el}(N/2-1) > \cdots > E_{\rm el}(1/2) \quad \text{or} \quad E_{\rm el}(0).
\]
Here, the minimum value in the above inequality corresponds to $E_{\rm el}(0)$ when $N$ is even, and $E_{\rm el}(1/2)$ when $N$ is odd.
\end{Thm}
From this theorem, it can be concluded that in a one-dimensional many-electron system with a broad class of potentials, the ground state with total spin $S$ is energetically more stable for smaller values of $S$. Hence, the system exhibits antiferromagnetic tendencies.

\subsection{Ordering of Energy Levels in the Fröhlich Model}
Consider the one-dimensional {\bf Fr\"{o}hlich model with ultraviolet cutoff}:
\begin{align}
H_{\vepsilon}\deq H_{\rm el}+
\sum_{j=1}^N 
\phi_{\vepsilon}(x_j)
+N_{\rm ph}. \label{DefHFr}
\end{align}
Here, the interaction between the electrons and phonons is defined as follows:
\begin{align}
\phi_{\vepsilon}(x)&\deq 
 \left(
\frac{\sqrt{2}}{L} \alpha
\right)^{1/2}
\sum_{k\in D^*}\ex^{-\vepsilon k^2}\left(
\ex^{\im kx} a_k+\ex^{-\im kx}a_k^*
\right). \notag
\end{align}
The parameter $\vepsilon\in (0, \infty)$ represents the ultraviolet cutoff. 
The operator $\He$ acts on the following Hilbert space:
\be
\h\deq \mathfrak{E}\otimes \F. \notag
\ee
Here, $\F$ denotes the bosonic Fock space over $\ell^2(D^*) \ (D^*\deq \frac{\pi}{L}\mathbb{Z})$:
\be
\F\deq \bigoplus_{n=0}^{\infty} \otimes^n_{\rm s} \ell^2(D^*). \notag
\ee
The operators \(a_k\) and \(a_k^*\) denote the annihilation and creation operators, respectively, and satisfy the commutation relations
\[
[a_k, a_{k'}^*] = \delta_{k,k'}, \quad [a_k, a_{k'}] = 0.
\]
More precisely, these  relations hold on a dense subspace of \(\F\), for example, on the finite-particle space \(\F_{\rm fin}\). Here, \(\F_{\rm fin} \deq \hat{\bigoplus}_{n=0}^{\infty} \otimes^n_{\rm s} \ell^2(D^*)\), where \(\hat{\oplus}\) denotes the incomplete direct sum. The operator \(N_{\rm ph}\) is the number operator for the bosons, and on the finite-particle space \(\F_{\rm fin}\) it is expressed as
\[
N_{\rm ph} = \sum_{k \in D^*} a_k^* a_k.
\]
For a detailed treatment of the creation and annihilation operators, as well as the second quantization operator, see, for example, \cite{Arai2023, Bratteli1997}.
By the Kato--Rellich theorem \cite{Reed1975}, for any $\vepsilon\in (0, \infty)$, $\He$ is a lower semi-bounded self-adjoint operator. 

When $\vepsilon=0$, i.e.,  the ultraviolet cutoff is removed, the operator $\phi_{\vepsilon=0}(x)$ is ill-defined. However, by applying Nelson\rq{}s arguments  \cite{Nelson1964}, the following fact holds: 
There exists a lower semi-bounded self-adjoint operator $H_{\rm F}$ such that as $\vepsilon \to +0$, 
$H_{\vepsilon}$ converges to $H_{\rm F}$ in the strong  resolvent sense. 
$H_{\rm F}$ is referred to as the {\bf Fr\"{o}hlich Hamiltonian without the ultraviolet cutoff}, or simply the Fr\"{o}hlich Hamiltonian.

Next, we investigate the existence of a ground state for the operator \( H_{\rm F} \). Given that \( H_{\rm el} \) has a discrete spectrum, it might initially appear evident that \( H_{\rm F} \) should also possess a ground state. Nevertheless, upon closer scrutiny of the spectrum of \( N_{\rm ph} \), given by \( \mathrm{spec}(N_{\rm ph}) = \{0\} \cup \mathbb{N} \), one observes that all eigenvalues, except for zero, belong to the essential spectrum. Furthermore, as the operator \( H_{\rm F} \) is defined via an ultraviolet cutoff removal limit, its explicit form remains elusive. Consequently, establishing the existence of a ground state for \(H_{\rm F}\) is a considerably more subtle and challenging task than might first be anticipated.

In addition to \hyperlink{H1}{\bf (H. 1)}, we assume the following condition:  
\begin{description}  
\item[\hypertarget{H2}{\bf (H. 2)}]  
The potential \( U \) is bounded from below, i.e., $\displaystyle \inf_{\bs{x}\in D^N} U(\bs{x}) > -\infty$.  
\end{description}
Throughout the remainder of this paper, we consistently assume \hyperlink{H1}{\bf (H. 1)} and \hyperlink{H2}{\bf (H. 2)}. However, it should be noted that \hyperlink{H2}{\bf (H. 2)} is not optimal and may potentially be replaced by a more general hypothesis.

To state the results, let us define
\be
\h^S \deq \mathfrak{E}^S \otimes \F. \notag
\ee
\begin{Thm}\label{ExGs}
For every \(\alpha > 0\) and each \(S \in \mathbb{S}_N\), \(H_{\mathrm{F}}\) admits a ground state in each \(\h^S\), and it is unique apart from the trivial \((2S+1)\)-fold degeneracy. 
\end{Thm}

The proof of Theorem \ref{ExGs} will be provided in Section \ref{PfExGs}. 
We emphasize that, despite the fermionic nature of electrons, the uniqueness of the ground state holds due to the special properties of one-dimensional systems.

For each $S\in \mathbb{S}_N$, we  define
\[
E(S) \deq \inf \mathrm{spec}(H_{\rm F} \restriction \h^S).
\]
The following theorem indicates that Mattis--Lieb theorem still holds even in the presence of electron-phonon interactions:
\begin{Thm}\label{EnOrd}
For any $\alpha>0$, the following inequality holds:
\be
E(N/2) > E(N/2-1) > \cdots > E(1/2) \quad \text{or} \quad E(0). \notag
\ee
Here, the minimum value in the above inequality corresponds to $E(0)$ when $N$ is even, and $E(1/2)$ when $N$ is odd.
\end{Thm}
The proof of Theorem \ref{EnOrd} will be provided in Section \ref{PfofThmEnOrd}.

The following theorem provides insight into how the ground state energy $ E_{\rm el}(S) $ of the many electron system is influenced by the electron-phonon interaction:

\begin{Thm}\label{GEBasic}
For any $ S \in \mathbb{S}_N $, the following hold:
\begin{itemize}
\item[\rm (i)] For all $ \alpha>0 $, $ E_{\rm el}(S) > E(S) $.
\item[\rm (ii)] $ E(S) $ is monotonically non-increasing with respect to $ \alpha $.
\item[\rm (iii)] $ \displaystyle \lim_{\alpha \to +0} E(S) = E_{\rm el}(S) $.
\end{itemize}
\end{Thm}

\begin{Rem}\rm
\begin{itemize}
\item For $ \vepsilon > 0 $, the inequality corresponding to (i) can be proven relatively easily. The crucial point of this assertion is that this strict inequality holds when $ \vepsilon = 0 $.
\item The key points of assertion (ii) are that the statement pertains to the Hamiltonian with the ultraviolet cutoff removed, and that it holds even considering the fermionic statistics of the electrons.
\end{itemize}
\end{Rem}

Theorem \ref{GEBasic} will be proved in Section \ref{PfofThmGEBasic}.

\section{Structure of $M$-subspaces}\label{SecMsubspace}

\subsection{Preliminaries}
In this section, we introduce several insightful representations of the Hilbert space for an $N$-electron system. These representations encapsulate the unique characteristics of one-dimensional systems and are crucial for the proofs of the main theorems. These representations will be utilized in the subsequent sections without additional elaboration.

\subsubsection*{Useful Identification 1}

Consider the Hilbert space of single-particle states without spin on the  domain $ D $, which is given by $ L^2(D) $. For an $ N $-particle system, the corresponding Hilbert space is $ \bigotimes^N L^2(D) $. A natural identification between this tensor product and the space $ L^2(D^N) $ is established via the relation:
\begin{align}
\bigotimes^N L^2(D) = L^2(D^N), \quad
f_1 \otimes \cdots \otimes f_N = f_1 \times \cdots \times f_N, \quad (f_1, \dots, f_N \in L^2(D)), \label{TensorIdn}
\end{align}
where the product function is defined by $ (f_1 \times \cdots \times f_N)(x_1, \dots, x_N) = f_1(x_1) \cdots f_N(x_N) $.
Under the identification \eqref{TensorIdn}, the action of $ S_{\pi} $ corresponds to
\[
\left(S_{\pi} (f_1 \otimes \cdots \otimes f_N)\right)(x_1, \dots, x_N) = f_1\left(x_{\pi^{-1}(1)}\right) \cdots f_N\left(x_{\pi^{-1}(N)}\right).
\]

Let $ A_N $ be the anti-symmetrizer on $ \bigotimes^N L^2(D) $, defined by
\[
A_N \deq  \frac{1}{N!} \sum_{\pi \in \mathfrak{S}_N} \operatorname{sgn}(\pi) S_{\pi}.
\]
For $ f_1, \dots, f_N \in L^2(D) $, the anti-symmetrized state is given by
\[
f_1 \wedge \cdots \wedge f_N \deq A_N(f_1 \otimes \cdots \otimes f_N).
\]
This leads to the following expression:
\begin{align*}
(f_1 \wedge \cdots \wedge f_N)(x_1, \dots, x_N)
&= \frac{1}{N!} \sum_{\pi \in \mathfrak{S}_N} \operatorname{sgn}(\pi) \, f_1\left(x_{\pi^{-1}(1)}\right) \cdots f_N\left(x_{\pi^{-1}(N)}\right) \\
&= \frac{1}{N!} \sum_{\pi \in \mathfrak{S}_N} \operatorname{sgn}(\pi) \, f_1\left(x_{\pi(1)}\right) \cdots f_N\left(x_{\pi(N)}\right).
\end{align*}
Here, we have used the property $ \operatorname{sgn}(\pi^{-1}) = \operatorname{sgn}(\pi) $.

In general, the anti-symmetrization of $ \varPhi \in L^2(D^N) $ is given by
\[
(A_N \varPhi)(x_1, \dots, x_N) = \frac{1}{N!} \sum_{\pi \in \mathfrak{S}_N} \operatorname{sgn}(\pi) \, \varPhi\left(x_{\pi(1)}, \dots, x_{\pi(N)}\right).
\]

\subsubsection*{Useful Identification 2}
The Hilbert space describing the interacting system of particles and the Bose field is given by  
\[  
\K \deq  L^2(D^N) \otimes \F.  
\]  
This space can be identified with the space of $\F$-valued $L^2$-functions:  
\be
\K =L^2(D^N; \F)= \int^{\oplus}_{D^N} \F \, d\bs{x}.  \label{KFiber}
\ee
Here, the right-hand side denotes the fiber direct integral of $\F$.\footnote{For details on fiber direct integrals of Hilbert spaces, see \cite[Chapter XIII.16]{Reed1978}.}

\subsubsection*{Useful Identification 3}

To detail the structure of the Hilbert space for an $ N $-electron system interacting with phonons, we emphasize the essential property that electrons possess spin-$ \frac{1}{2} $. The Hilbert space for a single-electron state is given by $ L^2(\vLa) $, with the fundamental identification:
\[
L^2(\vLa) = L^2(D) \oplus L^2(D) = L^2(D) \otimes \mathbb{C}^2.
\]
Building on this, the following identifications can be made:
\begin{align}
\left( \bigotimes^N L^2(\vLa) \right) \otimes \F 
&= \K \otimes \left( \bigotimes^N \mathbb{C}^2 \right), \notag\\
\left\{ (f_1 \otimes \eta_{\sigma_1}) \otimes \cdots \otimes (f_N \otimes \eta_{\sigma_N}) \right\} \otimes \psi 
&= \left\{ \left(f_1 \otimes \cdots \otimes f_N \right) \otimes \psi \right\} \otimes \eta_{\bs{\sigma}}, \label{IdnInH}
\end{align}
where $ f_1, \dots, f_N \in L^2(D) $ and $ \psi \in \F $.
Moreover, the anti-symmetrization of states leads to the following relation:
\begin{align}
(f_1 \otimes \eta_{\sigma_1}) \wedge \cdots \wedge (f_N \otimes \eta_{\sigma_N}) \otimes \psi 
= \sum_{\pi \in \mathfrak{S}_N} \frac{\mathrm{sgn}(\pi)}{N!} \left\{ S_{\pi} \left(f_1 \otimes \cdots \otimes f_N \right) \otimes \psi \right\} \otimes 
(s_{\pi} \eta_{\bs{\sigma}}). \label{IdnInH2}
\end{align}

\subsection{Definition of $\K_{{\rm as}, p}$}

For each $M \in \{-N/2, -N/2+1, \dots, N/2\}$, define  
\be
p \deq \frac{N}{2} - M. \label{Defp}
\ee
As illustrated in Subsection \ref{Mspace}, \(M\) corresponds to the eigenvalues of \(S^{(3)}\).
The variable \( p \) takes on the values \( 0, 1, \dots, N \).
Consider $\vPsi \in L^2(D^N)$. If $\vPsi(x_1, \dots, x_N)$ is anti-symmetric with respect to $x_1, \dots, x_p$ and also anti-symmetric with respect to $x_{p+1}, \dots, x_N$, we write it explicitly as  
\[
\vPsi(x_1, \dots, x_p\, |\, x_{p+1}, \dots, x_N),
\]  
to emphasize this property.\footnote{This notation is borrowed from \cite{Lieb1962}.}  For each $M$, we define a closed subspace of $L^2(D^N)$  as follows:  
\[
L^2_{{\rm as}, p}(D^N) \deq  \left\{ \vPsi \in L^2(D^N) : \vPsi = \vPsi(x_1, \dots, x_p\, |\, x_{p+1}, \dots, x_N) \right\}.
\]

Define  
\[  
\K_{{\rm as}, p} \deq  L^2_{{\rm as}, p}(D^N) \otimes \F.  
\]  
Under the  identification \eqref{KFiber}, each element $\vPsi \in \K_{{\rm as}, p}$ is anti-symmetric with respect to $x_1, \dots, x_p$ and also anti-symmetric with respect to $x_{p+1}, \dots, x_N$.  
Using the same notation as above, this is expressed as  
\[  
\vPsi = \vPsi(x_1, \dots, x_p\, |\, x_{p+1}, \dots, x_N).  
\]

\subsection{Representation of Vectors in the $M$-subspace}\label{Mspace}
For each \(M \in \mathrm{spec}(S^{(3)})=\{-N/2, -N/2+1, \dots, N/2\}\), the {\bf \(M\)-subspace of} \(\E\) is defined as
\[
\E_M \coloneqq \ker\bigl(S^{(3)} - M\bigr) \cap \E.
\]
The \(M\)-subspace has proven instrumental in analyzing the magnetic properties of many-electron systems (see, e.g., \cite{Mattis2006}). In this subsection, we examine the properties of the \(M\)-subspace.

 It follows from Eq. \eqref{IdnInH2}  that any element $\vPhi$ of $\E_M$ can be expressed as:
\be
\vPhi = \sum_{\pi \in \mathfrak{S}_N} \frac{\mathrm{sgn}(\pi)}{N!} S_{\pi} \vPsi \otimes s_{\pi} \eta_{\bs{\sigma}_M}, \quad
\vPsi \in L^2(D^N). \label{BasicVec1}
\ee
Here, the reference spin configuration $\bs{\sigma}_M$ is defined by:
\be
\bs{\sigma}_M \deq (\underbrace{-1, \dots, -1}_p, \underbrace{+1, \dots, +1}_{N-p}). \notag
\ee

Consider the $M$-subspace of $\h$:
\be
\h_M \deq \E_M \otimes \F. \notag
\ee
Using Eqs.   \eqref{IdnInH} and \eqref{BasicVec1}, any element $\vPhi$ of $\h_M$ can be represented as:
\be
\vPhi = \sum_{\pi \in \mathfrak{S}_N} \frac{\mathrm{sgn}(\pi)}{N!} S_{\pi} \vPsi \otimes s_{\pi} \eta_{\bs{\sigma}_M}, \quad
\vPsi \in \K. \label{BasicVec}
\ee

Define a subgroup of the symmetric group $\mathfrak{S}_N$ as:
\be
\mathfrak{S}_{N, \rm T}^{(M)} \deq  \{\tau \times \xi : \tau \in \mathfrak{S}_p,\, \xi \in \mathfrak{S}_{N-p}\}. \notag
\ee
Here, $p$ is defined in Eq. \eqref{Defp};  for any $\tau \in \mathfrak{S}_p$ and $\xi \in \mathfrak{S}_{N-p}$, we define $\tau \times \xi \in \mathfrak{S}_N$ by
\be
(\tau \times \xi)(1, \dots, N) \deq \left( \tau(1), \dots, \tau(p),\, \xi(p+1), \dots, \xi(N) \right). \notag
\ee

Let $q = \min\{p, N - p\}$. For $n = 1, \dots, q$, define the sets consisting of products of transpositions:
\begin{align*}
\varXi_n \deq  \Big\{ (i_1, j_1) \circ \dots \circ (i_n, j_n) \in \mathfrak{S}_N \setminus \mathfrak{S}_{N, \rm T}^{(M)} :\ & i_{\alpha} \neq i_{\beta},\ j_{\alpha} \neq j_{\beta}\ \text{if}\ \alpha \neq \beta,\no
& i_1, \dots, i_n \in \{1, \dots, p\},\ j_1, \dots, j_n \in \{p+1, \dots, N\} \Big\}.
\end{align*}
Here, $(i, j)$ denotes the permutation (transposition) swapping $i$ and $j$. Elements of $\varXi_n$ swap $n$ elements selected from each of the two clusters $\{1, \dots, p\}$ and $\{p+1, \dots, N\}$. Then, define
\be
\mathfrak{S}_{N, \rm NT}^{(M)} \deq \bigsqcup_{n=1}^{q} \varXi_n. \notag
\ee
It is readily observed that
\be
\mathfrak{S}_{N, \rm NT}^{(M)} \circ \mathfrak{S}_{N, \rm T}^{(M)} = \mathfrak{S}_N \setminus \mathfrak{S}_{N, \rm T}^{(M)}. \label{Prod}
\ee

Given the vector $\vPhi$ defined in Eq. \eqref{BasicVec}, we define its \textbf{representative part} $\mathscr{R}_M(\vPhi) \in \mathfrak{K}_{\mathrm{as}, p}$ by
\be
\mathscr{R}_M(\vPhi) = \binom{N}{p}^{-1} A_p \otimes A_{N-p} \vPsi. \notag
\ee
For clarity, note that $\vPsi \in \mathfrak{K}$ is the vector appearing on the right-hand side of Eq.  \eqref{BasicVec}. Observe that $\mathscr{R}_M(\vPhi)$ is not fully anti-symmetrized; hence, $\mathscr{R}_M(\vPhi) \notin \h_M$.

\begin{Prop}\label{PropR_M}
Any \(\vPhi \in \h_M\) can be written as 
\be
\vPhi 
= \mathscr{R}_M(\vPhi)\,\otimes\,\eta_{\bs{\sigma}_M}
\;+\;
\sum_{\pi \in \mathfrak{S}_{N, \mathrm{NT}}^{(M)}} \mathrm{sgn}(\pi)\,\bigl(S_{\pi}\,\otimes\,s_{\pi}\bigr)\bigl(\mathscr{R}_M(\vPhi)\,\otimes\,\eta_{\bs{\sigma}_M}\bigr).
\label{PhiDec2}
\ee
Conversely, given any \(\vPsi \in \K_{\mathrm{as}, p}\), define \(\vPhi \in \h_M\) by
\be
\vPhi 
\deq \vPsi \,\otimes\,\eta_{\bs{\sigma}_M}
\;+\;
\sum_{\pi \in \mathfrak{S}_{N, \mathrm{NT}}^{(M)}} \mathrm{sgn}(\pi)\,\bigl(S_{\pi}\,\otimes\,s_{\pi}\bigr)\bigl(\vPsi \,\otimes\,\eta_{\bs{\sigma}_M}\bigr),
\label{PsiDec}
\ee
then \(\mathscr{R}_M(\vPhi)=\binom{N}{p}\,\vPsi\). Consequently, the linear map \(\mathscr{R}_M: \h_M \to \mathfrak{K}_{\mathrm{as}, p}\) is bijective.
\end{Prop}
\begin{Rem}
Note that each term on the right-hand side of Eq. \eqref{PhiDec2} is not an element of $\h_M$. Therefore, Eq. \eqref{PhiDec2} should be regarded as an equality in the larger Hilbert space $\left( \bigotimes^N L^2(\vLa) \right) \otimes \F$.
\end{Rem}
\begin{proof}
Assuming $\vPhi$ is given by Eq. \eqref{BasicVec}.  Since $s_{\pi} \eta_{\bs{\sigma}_M} = \eta_{\bs{\sigma}_M}$  for $\pi \in \mathfrak{S}_{N, \rm T}^{(M)}$, we have
\begin{align}
\vPhi &= \sum_{\pi \in \mathfrak{S}_{N, \rm T}^{(M)}} \frac{\mathrm{sgn}(\pi)}{N!} S_{\pi} \vPsi \otimes \eta_{\bs{\sigma}_M} + \sum_{\pi \in \mathfrak{S}_N \setminus \mathfrak{S}_{N, \rm T}^{(M)} } \frac{\mathrm{sgn}(\pi)}{N!} S_{\pi} \vPsi \otimes s_{\pi} \eta_{\bs{\sigma}_M}. \label{PhiDec}
\end{align}
When decomposed in this manner, the first term on the right-hand side equals $\mathscr{R}_M(\vPhi) \otimes \eta_{\bs{\sigma}_M}$. On the other hand, using Eq. \eqref{Prod}, we obtain
\begin{align*}
\sum_{\pi \in \mathfrak{S}_N \setminus \mathfrak{S}_{N, \rm T}^{(M)} } \frac{\mathrm{sgn}(\pi)}{N!} S_{\pi} \otimes s_{\pi} &= \sum_{\pi_1 \in \mathfrak{S}_{N, \rm NT}^{(M)},\ \pi_2 \in  \mathfrak{S}_{N, \rm T}^{(M)}} \frac{\mathrm{sgn}(\pi_1 \circ \pi_2)}{N!} S_{\pi_1 \circ \pi_2} \otimes s_{\pi_1 \circ \pi_2} \no
&= \left( \sum_{\pi \in \mathfrak{S}_{N, \rm NT}^{(M)}} \mathrm{sgn}(\pi) S_{\pi} \otimes s_{\pi} \right) \left( \sum_{\pi \in \mathfrak{S}_{N, \rm T}^{(M)}} \frac{\mathrm{sgn}(\pi)}{N!} S_{\pi} \otimes s_{\pi} \right).
\end{align*}
Therefore, the second term on the right-hand side of Eq. \eqref{PhiDec} equals $\sum_{\pi \in \mathfrak{S}_{N, \rm NT}^{(M)}} \mathrm{sgn}(\pi) (S_{\pi} \otimes s_{\pi})(\mathscr{R}_M(\vPhi) \otimes \eta_{\bs{\sigma}_M})$. The proof of the remaining assertions is straightforward.
\end{proof}

The subspace of $\h$:
\be
\h^S_M \deq \h_M \cap \h^S = \left\{
\vPsi \in \h\, :\, S^{(3)} \vPsi = M \vPsi,\quad \bs{S}^2 \vPsi = S(S+1) \vPsi
\right\}
\ee
is important in the subsequent discussion.

\begin{Coro}\label{IffMM}
For any $\vPhi \in \h_M$, the following statements hold:
\begin{itemize}
\item[\rm (i)] $\vPhi = 0$ if and only if $\mathscr{R}_M(\vPhi) = 0$.
\item[\rm (ii)] $\mathscr{R}_M(\vPhi) = \left( \mathscr{R}_M(\vPhi) \right)(x_1, \dots, x_p\,  |\,  x_{p+1}, \dots, x_N)$. I.e., $\mathscr{R}_M(\vPhi)\in \mathfrak{K} _{\mathrm{as}, p}$.
\item[\rm (iii)] Let $S^{(\pm)} = S^{(1)} \pm \im S^{(2)}$ denote the raising and lowering operators. Then,
\begin{align}
\mathscr{R}_{M+1}(S^{(+)} \vPhi) &=
\mathscr{R}_M(\vPhi) - \sum_{j=p+1}^N S_{(p, j)} \mathscr{R}_M(\vPhi), \label{RS+}\\
\mathscr{R}_{M+1}(S^{(-)} \vPhi) &=
\mathscr{R}_M(\vPhi) - \sum_{j=1}^p S_{(j, p+1)} \mathscr{R}_M(\vPhi). \label{RS-}
\end{align}
In particular, for any $M\in \mathrm{spec}(S^{(3)})$ with $M\ge 0$, a necessary and sufficient condition for $\vPhi \in \h^M_M$ (respectively, $\vPhi\in \h_{-M}^M$)is that the right-hand side of Eq. \eqref{RS+} (respectively,  Eq. \eqref{RS-}) equals zero.
\end{itemize}
\end{Coro}
\begin{proof}
(i) and (ii) are evident.

(iii)
For $f \otimes \eta_{\bs{\sigma}_{M+1}}$, we refer to $f \in \K$ as the {\bf coefficient of} $\eta_{\bs{\sigma}_{M+1}}$. Applying Eq. \eqref{PhiDec2}, $ S^{(+)} \vPhi $ can be decomposed as
\[
S^{(+)} \vPhi = f \otimes \eta_{\bs{\sigma}_{M+1}} + \text{additional terms}.
\]
In this representation, the coefficient $ f $ associated with $ \eta_{\bs{\sigma}_{M+1}} $ is precisely $ \mathscr{R}_{M+1}(S^{(+)} \vPhi) $.

Let $ \sigma^{(+)} = \sigma^{(1)} + \im \sigma^{(2)} $. Since $ S^{(+)} = \sum_{j=1}^N \sigma^{(+)}_j $, it follows that
\[
S^{(+)} \mathscr{R}_M(\vPhi) \otimes \eta_{\bs{\sigma}_M} = \sum_{j=1}^N \mathscr{R}_M(\vPhi) \otimes \sigma^{(+)}_j \eta_{\bs{\sigma}_M}.
\]
Among the terms appearing on the right-hand side of this equation, the coefficient associated with $\eta_{\bs{\sigma}_{M+1}}$ is $\mathscr{R}_M(\vPhi)$.

Next, denote the second term on the right-hand side of Eq.  \eqref{PhiDec2} by $ \vPhi_1 $. While $ S^{(+)} \vPhi_1 $ expands into multiple terms, collecting those of the form $ f \otimes \eta_{\bs{\sigma}_{M+1}} $ yields
\[
f = -\sum_{j=p+1}^N S_{(p, j)} \mathscr{R}_M(\vPhi).
\]
Here, $ (p, j) $ represents the transposition exchanging $ p $ and $ j $.
\end{proof}

\subsection{Structure of $\K_{{\rm as}, p}$ }

We investigate the structure of elements in $ L^2_{{\rm as}, p}(D^N) $. For this purpose, we define the domain  
\begin{align*}
D^N_{p, \neq} \deq \Big\{(x_1, \dots, x_N) \in D^N \ : \ & \text{$ x_i \neq x_j $ for any pair $ (i, j) \in \{1, \dots, p\}^{\times 2} $ with $ i \neq j $} \no
&\text{and for any pair $ (i, j) \in \{p+1, \dots, N\}^{\times 2} $ with $ i \neq j $} \Big\}.
\end{align*}
Due to the anti-symmetry property, we identify  
\be
L^2_{{\rm as}, p}(D^N) = L^2_{{\rm as}, p}(D^N_{p, \neq}), \label{AntiZero}
\ee
where
\begin{align}
L^2_{{\rm as}, p}(D^N_{p, \neq}) \deq \left\{\vPsi \in L^2(D^N_{p, \neq}) : \vPsi = \vPsi(x_1, \dots, x_p \mid x_{p+1}, \dots, x_N) \right\}. \notag
\end{align}

We now introduce the reference domain $ D^{N, (p)} $, defined as
\be
D^{N, (p)} \deq \left\{ (x_1, \dots, x_N) \in D^N : x_1 < x_2 < \dots < x_p,\ x_{p+1} < \dots < x_N \right\}. \label{DefDNp} \notag
\ee
For permutations $ \tau \in \mathfrak{S}_p $ (acting on $ \{1, \dots, p\} $) and $ \xi \in \mathfrak{S}_{N-p} $ (acting on $ \{p+1, \dots, N\} $), we define  
\be
D^{N, (p)}_{\tau, \xi} \deq  \left\{ (x_1, \dots, x_N) \in D^N : x_{\tau(1)} < x_{\tau(2)} < \dots < x_{\tau(p)},\ x_{\xi(p+1)} < \dots < x_{\xi(N)} \right\}, \notag
\ee
where $ D^{N, (p)}_{\mathrm{Id}, \mathrm{Id}} \deq  D^{N, (p)} $ (with $ \mathrm{Id} $ being the identity permutation). It follows that
\be
D^N_{p, \neq} = \bigsqcup_{\tau \in \mathfrak{S}_p, \xi \in \mathfrak{S}_{N-p}} D^{N, (p)}_{\tau, \xi}. \notag
\ee

The properties of each $ \vPsi \in L^2_{{\rm as}, p}(D^N) $ are determined by its behavior on $ D^{N, (p)} $. Specifically, letting $ \vPsi_{D^{N, (p)}} = \vPsi \mathbbm{1}_{D^{N, (p)}} $, we have
\be
\vPsi = \sum_{\tau \in \mathfrak{S}_p, \xi \in \mathfrak{S}_{N-p}} \mathrm{sgn}(\tau)\mathrm{sgn}(\xi) S_{\tau^{-1} \times \xi^{-1}} \vPsi_{D^{N, (p)}}. \label{ExtP0}
\ee
Here, we used the identification   \eqref{AntiZero}, and the operator $ S_{\tau^{-1} \times \xi^{-1}} $ is defined by
\[
\left(S_{\tau^{-1} \times \xi^{-1}} \varPhi \right)(x_1, \dots, x_N) = \varPhi\left(x_{\tau(1)}, \dots, x_{\tau(p)}, x_{\xi(p+1)}, \dots, x_{\xi(N)}\right).
\]
Furthermore, for a given set \(A\), we denote by \(\mathbbm{1}_A\) the indicator function of \(A\).
Conversely, if $ \vPsi_{D^{N, (p)}} $ belongs to $ L^2(D^{N, (p)}) $, then $ \vPsi_{D^{N, (p)}} $ can be extended to an element of $ L^2_{{\rm as}, p}(D^N) $ via Eq. \eqref{ExtP0}.

We now define  
\be
\K_{D^{N, (p)}} \deq L^2(D^{N, (p)}) \otimes \F. \label{DefKDNP}
\ee
The properties described above naturally extend to $ \K_{D^{N, (p)}} $, as summarized in the following lemma:

\begin{Lemm}\label{PsiDecD}
For each \(\vPsi \in \K_{\mathrm{as}, p}\), set \(\vPsi_{D^{N,(p)}} = \vPsi \,\mathbbm{1}_{D^{N,(p)}}\). Then
\be
\vPsi = \sum_{\tau \in \mathfrak{S}_p, \xi \in \mathfrak{S}_{N-p}} \mathrm{sgn}(\tau)\mathrm{sgn}(\xi) S_{\tau^{-1} \times \xi^{-1}} \vPsi_{D^{N, (p)}}. \label{ExtP}
\ee
Conversely, any \(\vPsi_{D^{N,(p)}} \in \K_{D^{N,(p)}}\) extends to \(\K_{\mathrm{as}, p}\) via this formula.

 Define the linear map \(\iota_p : \K_{D^{N,(p)}} \to \K_{\mathrm{as}, p}\) by
\[
\iota_p\,\vPsi_{D^{N,(p)}} \deq  \sum_{\tau \in \mathfrak{S}_p,\, \xi \in \mathfrak{S}_{N-p}} \mathrm{sgn}(\tau)\,\mathrm{sgn}(\xi)\, S_{\tau^{-1} \times \xi^{-1}}\, \vPsi_{D^{N,(p)}}.
\]
Then, \(\|\iota_p\,\vPsi_{D^{N,(p)}}\|^2 = p!\,(N-p)!\,\|\vPsi_{D^{N,(p)}}\|^2\), and \(\iota_p\) is bijective.
\end{Lemm}

\subsection{Characteristics of the Ground State}

To clarify the essential conditions, we consider a general Hamiltonian. Let $ K $ be a non-negative self-adjoint operator on $ \K $ satisfying the following condition:

\begin{description}
\item[\hypertarget{A1}{(A. 1)}] 
For any $ \pi \in \mathfrak{S}_N $, it holds that $ S_{\pi} \D(K) \subseteq \D(K) $ and $ S_{\pi} K = K S_{\pi} $ on $ \D(K) $.
\item[\hypertarget{A2}{(A. 2)}]
Let $P^{(p)}$ denote the orthogonal projection operator from $\mathfrak{K}$ onto its closed subspace $\mathfrak{K}_{D^{N, (p)}}$.  
 For all $M \in \{-N/2, -N/2 + 1, \dots, N/2\}$, the operators $K$ and $P^{(p)}$ commute in the strong sense:  
$
\mathrm{e}^{-\beta K} P^{(p)} = P^{(p)} \mathrm{e}^{-\beta K}$ for all  $\beta > 0$.
\end{description}

As we shall see in the following sections, \(K\) corresponds to the Hamiltonian \(H_{\rm F}\).

For a lower semi-bounded self-adjoint operator \(A\), let \(\mathcal{E}(A) \deq \inf \mathrm{spec}(A)\).
In what follows, we adopt the following convention: if \(\mathcal{E}(A)\) is an eigenvalue, then we say that \(A\) admits a {\bf ground state}. Naturally, this terminology is employed only when \(A\) represents the energy operator of a particular physical system.

\begin{Lemm}\label{RPartGS}
Let $\mathsf{K}_M$ be the restriction of $K \otimes \one_{\otimes^N \mathbb{C}^2}$ to $\h_M$, which is a non-negative self-adjoint operator. Suppose $\vPhi \in \h_M$. If $\mathscr{R}_M(\vPhi)$ is a ground state of $K \restriction \K_{{\rm as}, p}$, then $\vPhi$ is a ground state of $\mathsf{K}_M$. Conversely, if $\vPhi$ is a ground state of $\mathsf{K}_M$, then $\mathscr{R}_M(\vPhi)$ is a ground state of $K \restriction \K_{{\rm as}, p}$. In both cases, $\mathcal{E}(\mathsf{K}_M) = \mathcal{E}(K \restriction \K_{{\rm as}, p})$. Furthermore, if the ground state of $K \restriction \K_{{\rm as}, p}$ is unique, then the ground state of $\mathsf{K}_M$ is also unique.
\end{Lemm}

\begin{proof}
Using \hyperlink{A1}{\bf (A.1)}, we observe that for any $\vPhi \in \D(\mathsf{K}_M)$,
\[
\langle \vPhi \mid \mathsf{K}_M \vPhi \rangle
= \bigl(1 + \sharp\mathfrak{S}_{N, \mathrm{NT}}^{(M)}\bigr)
\bigl\langle \mathscr{R}_M(\vPhi)\, \big\vert\,  K \mathscr{R}_M(\vPhi)\bigr\rangle,
\]
which directly implies $\mathcal{E}(\mathsf{K}_M) = \mathcal{E}(K \restriction \K_{\mathrm{as}, p})$.
Here, for any given set \( S \), the notation \( \sharp S \) denotes its cardinality.

Let us denote by $\vPhi_1$ the second term on the right-hand side of Eq. \eqref{PhiDec2}. For any $\pi \in \mathfrak{S}_{N, \mathrm{NT}}^{(M)}$, we have $\langle \eta_{\bs{\sigma}_M} \mid S_{\pi} \eta_{\bs{\sigma}_M} \rangle = 0$, which shows $\langle \mathscr{R}_M(\vPhi) \otimes \eta_{\bs{\sigma}_M} \mid K \otimes \one \,\vPhi_1 \rangle = 0$. Consequently, if $\vPhi$ is a ground state of $\mathsf{K}_M$, then
\[
\mathcal{E}(\mathsf{K}_M)\|\mathscr{R}_M(\vPhi)\|^2
= \mathcal{E}(\mathsf{K}_M) \langle \mathscr{R}_M(\vPhi) \otimes \eta_{\bs{\sigma}_M} \mid \vPhi \rangle
= \langle \mathscr{R}_M(\vPhi) \mid K\,\mathscr{R}_M(\vPhi)\rangle.
\]
Since the left-hand side equals $\mathcal{E}(K \restriction \K_{\mathrm{as}, p}) \|\mathscr{R}_M(\vPhi)\|^2$, $\mathscr{R}_M(\vPhi)$ must be a ground state of $K \restriction \K_{\mathrm{as}, p}$. Conversely, if $\mathscr{R}_M(\vPhi)$ is a ground state of $K \restriction \K_{\mathrm{as}, p}$, then by again using \hyperlink{A1}{\bf (A.1)}, we obtain $K \otimes \one \,\vPhi_1 = \mathcal{E}(K \restriction \K_{\mathrm{as}, p}) \vPhi_1$. Therefore,
\[
\mathsf{K}_M \vPhi
= \mathcal{E}(K \restriction \K_{\mathrm{as}, p}) \vPhi
= \mathcal{E}(\mathsf{K}_M) \vPhi,
\]
hence $\vPhi$ is a ground state of $\mathsf{K}_M$. The uniqueness claim follows from the bijectivity of $\mathscr{R}_M$ from $\h_M$ onto $\K_{\mathrm{as}, p}$ (Proposition \ref{PropR_M}).
\end{proof}

\begin{Prop}\label{GSEng2}
Let \(K^{(p)} \deq  K \restriction \K_{D^{N,(p)}}\). Then
\(K^{(p)}\) possesses a ground state if and only if \(\mathsf{K}_M\) also does.  Moreover, if the ground state of \(K^{(p)}\) is unique, the same property holds for \(\mathsf{K}_M\). The representative part of any ground state of \(\mathsf{K}_M\) is specified by Eq. \eqref{ExtP}, employing a ground state \(\vPsi_{D^{N,(p)}}\) of \(K^{(p)}\). Furthermore, \(\mathcal{E}(\mathsf{K}_M) = \mathcal{E}(K^{(p)})\).
\end{Prop}

\begin{proof}
For any $\vPsi \in \D\bigl(K \restriction \K_{\mathrm{as}, p}\bigr)$, employing the representation \eqref{ExtP} and \hyperlink{A1}{\bf (A.1)}, we obtain
$
K \vPsi 
= \sum_{\tau \in \mathfrak{S}_p,\, \xi \in \mathfrak{S}_{N-p}} \mathrm{sgn}(\tau)\,\mathrm{sgn}(\xi)\,S_{\tau^{-1} \times \xi^{-1}}\,K^{(p)} \vPsi_{D^{N,(p)}},
$
which leads to
\[
\langle \vPsi \mid K \vPsi \rangle
= p! \,(N-p)!\,\bigl\langle \vPsi_{D^{N,(p)}} \bigm| K^{(p)} \vPsi_{D^{N,(p)}} \bigr\rangle.
\]
Hence, it follows that 
$
\mathcal{E}\bigl(K \restriction \K_{\mathrm{as}, p}\bigr)
= \mathcal{E}\bigl(K^{(p)}\bigr).
$

Let $\vPsi_{D^{N,(p)}}$ be a ground state of $K^{(p)}$, and define $\vPsi \in \K_{\mathrm{as}, p}$ via Eq.  \eqref{ExtP}. By applying \hyperlink{A1}{\bf (A.1)}, we see 
$
K \vPsi 
= \mathcal{E}\bigl(K^{(p)}\bigr) \,\vPsi 
= \mathcal{E}\bigl(K \restriction \K_{\mathrm{as}, p}\bigr)\,\vPsi,
$
so $\vPsi$ is a ground state of $K \restriction \K_{\mathrm{as}, p}$. Conversely, if $\vPsi$ is a ground state of $K \restriction \K_{\mathrm{as}, p}$, then from  Eq. \eqref{ExtP},
$$
\mathcal{E}\bigl(K \restriction \K_{\mathrm{as}, p}\bigr) \vPsi 
= K \vPsi 
= \sum_{\tau,\,\xi} \mathrm{sgn}(\tau)\,\mathrm{sgn}(\xi)\,S_{\tau^{-1} \times \xi^{-1}}\,K^{(p)} \vPsi_{D^{N,(p)}}.
$$
Therefore,
\[
\mathcal{E}\bigl(K \restriction \K_{\mathrm{as}, p}\bigr)
\|\vPsi_{D^{N,(p)}}\|^2
= \mathcal{E}\bigl(K \restriction \K_{\mathrm{as}, p}\bigr)\,\langle \vPsi_{D^{N,(p)}} \mid \vPsi \rangle
= \left\langle \vPsi_{D^{N,(p)}}\,  \Big|\,  K^{(p)} \vPsi_{D^{N,(p)}} \right\rangle.
\]
Since the left-hand side equals $\mathcal{E}\bigl(K^{(p)}\bigr)\|\vPsi_{D^{N,(p)}}\|^2$, $\vPsi_{D^{N,(p)}}$ must be a ground state of $K^{(p)}$.

The assertion regarding uniqueness follows from the bijectivity of the linear maps  \(\iota_p\) and $\mathscr{R}_M$ in Proposition \ref{PropR_M} and  Lemma \ref{PsiDecD}.
\end{proof}

\section{Proof of Theorem \ref{EnOrd}}\label{PfofThmEnOrd}

\subsection{Outline of the Proof}

We consider the {\bf Schrödinger representation} of $\mathfrak{F}$\footnote{For details on the Schrödinger representation, see, e.g.,  \cite{Lrinczi2020}.}:  
$
\mathfrak{F} = L^2(\mathcal{Q}, d\mu).
$
In this framework, we identify  
\be
\mathfrak{K}_{D^{N, (p)}} = L^2\left(D^{N, (p)} \times \mathcal{Q}, d\bs{x} \times d\mu\right),\label{KSch}
\ee  
where $\mathfrak{K}_{D^{N, (p)}}$ is defined in Eq. \eqref{DefKDNP}. This representation provides a concrete functional-analytic setting in which the properties of the operators can be analyzed effectively. In particular,  by utilizing the Schrödinger representation, we can effectively apply key properties of the heat semigroup, such as positivity preservation and, under appropriate conditions, positivity improvement. These properties play a crucial role in establishing the ordering of energy levels.

The third assumption, which is central to our proof, is stated as follows:

\begin{description}
\item[\hypertarget{A3}{(A. 3)}]  
For every \(p \in \{0, 1, \dots, N\}\), the following hold:
\begin{itemize}
\item \(K^{(p)}\) possesses a ground state. Recall that \(K^{(p)} = K \restriction \K_{D^{N,(p)}}\).
\item For all \(\beta > 0\), \(\mathrm{e}^{-\beta K^{(p)}}\) improves positivity in the Schrödinger representation.\footnote{More precisely, for any non-negative \(\vPsi \in L^2\bigl(D^{N, (p)} \times \mathcal{Q}, d\bs{x} \times d\mu\bigr)\setminus \{0\}\), \(\mathrm{e}^{-\beta K^{(p)}} \vPsi\) is stirctly positive for all \(\beta > 0\).}
\end{itemize}
\end{description}

Now,  consider a lower semi-bounded self-adjoint operator:
\[
\mathsf{K}\deq K\otimes \one_{\otimes^N\BbbC^2}\restriction \h,
\]  
which serves as an abstraction of the Hamiltonian of the electron-phonon interaction system described in Eq.  \eqref{DefHFr}. Utilizing the assumptions, we establish the following theorem, which characterizes the energy ordering of the eigenvalues:

\begin{Thm}\label{AbstEO}
For each $S \in \mathbb{S}_N$, let $E_{\sf K}(S) \deq \inf \mathrm{spec}(\mathsf{K} \restriction \h^S)$. Under the  assumptions \hyperlink{A1}{\bf (A. 1)}, \hyperlink{A2}{\bf (A. 2)}, and \hyperlink{A3}{\bf (A. 3)}, the following holds:
\be
E_{\mathsf{K}}(N/2) > E_{\mathsf{K}}(N/2-1) > \cdots > E_{\mathsf{K}}(1/2)\ \text{or}\  E_{\mathsf{K}}(0). \notag
\ee
Here, the minimum value in the above inequality corresponds to $E_{\mathsf{K}}(0)$ when $N$ is even, and $E_{\mathsf{K}}(1/2)$ when $N$ is odd.
\end{Thm}

The proof of Theorem \ref{AbstEO} is provided in Subsection \ref{PfofThmAbstEO}. Building on this result, we proceed to verify the implications for the specific case of Theorem \ref{EnOrd}.

\subsubsection*{Proof of Theorem \ref{EnOrd}, given Theorem \ref{AbstEO}}

To prove Theorem \ref{EnOrd}, we employ Theorem \ref{AbstEO} with \(K = H_{\rm F}\), the Hamiltonian in question. It suffices to verify that the assumptions \hyperlink{A1}{\bf (A.1)}, \hyperlink{A2}{\bf (A.2)}, and \hyperlink{A3}{\bf (A.3)} are fulfilled. These verifications proceed as follows:

First, \(H_{\rm F}\) trivially satisfies \hyperlink{A1}{\bf (A.1)} by virtue of its intrinsic symmetry properties. Next, \hyperlink{A2}{\bf (A.2)} is clearly satisfied. Finally, \hyperlink{A3}{\bf (A.3)} is established in Section \ref{PfExGs}: see Theorem \ref{Hyp2} for the existence of the ground state of \(K^{(p)}= H_{\rm F} \restriction \K_{D^{N,(p)}}\), and refer to Proposition \ref{PIEx} for the fact that \(\exp(-\beta K^{(p)})\) is positivity-improving.
\qed

\subsection{Proof of Theorem \ref{AbstEO}}\label{PfofThmAbstEO}

To prove Theorem \ref{AbstEO}, we begin by considering linearly independent vectors $\vphi_1, \dots, \vphi_N$ in $L^2(D)$. The anti-symmetric tensor product $\vphi_1 \wedge \cdots \wedge \vphi_p$ can be represented using the Slater determinant:
\be
(\vphi_1 \wedge \cdots \wedge \vphi_p)(x_1, \dots, x_p) = 
\begin{vmatrix}
\vphi_1(x_1) & \vphi_1(x_2) & \cdots & \vphi_1(x_p) \\
\vphi_2(x_1) & \vphi_2(x_2) & \cdots & \vphi_2(x_p) \\
\vdots & \vdots & & \vdots \\
\vphi_p(x_1) & \vphi_p(x_2) & \cdots & \vphi_p(x_p)
\end{vmatrix}. \notag
\ee

Now, take the vector $\vPhi \in \h_M$ whose representative part is given by
\begin{align}
&\mathscr{R}_M(\vPhi)(x_1, \dots, x_N)\no  
=& 
\left\{
\begin{vmatrix}
\vphi_1(x_1) & \vphi_1(x_2) & \cdots & \vphi_1(x_p) \\
\vphi_2(x_1) & \vphi_2(x_2) & \cdots & \vphi_2(x_p) \\
\vdots & \vdots & & \vdots \\
\vphi_p(x_1) & \vphi_p(x_2) & \cdots & \vphi_p(x_p)
\end{vmatrix} 
\times 
\begin{vmatrix}
\vphi_1(x_{p+1}) & \vphi_1(x_{p+2}) & \cdots & \vphi_1(x_N) \\
\vphi_2(x_{p+1}) & \vphi_2(x_{p+2}) & \cdots & \vphi_2(x_N) \\
\vdots & \vdots & & \vdots \\
\vphi_{N-p}(x_{p+1}) & \vphi_{N-p}(x_{p+2}) & \cdots & \vphi_{N-p}(x_N)
\end{vmatrix}
\right\} \varOmega. \label{SlVec}
\end{align}
Here, $\vOm$ is the Fock vacuum in $\F$.

\begin{Lemm}\label{SasekiProp}
Let $\vPhi \in \h_M$ be the vector whose representative part is given by  Eq. \eqref{SlVec}. Then $\vPhi$ belongs to $\h_M^M$. In particular, for the vector $\vPhi$ satisfying
\be
\mathscr{R}_M(\vPhi)(x_1, \dots, x_N) = \left\{ \prod_{1 \le i < j \le p} (x_j - x_i) \times \prod_{p+1 \le i < j \le N} (x_j - x_i) \right\} \varOmega, \label{Saseki}
\ee
we have that if \(M \ge 0\), then \(\vPhi \in \h_M^M\); whereas if \(M < 0\), then \(\vPhi \in \h_{M}^{-M}\). Moreover, in the Schrödinger representation, $\mathscr{R}_M(\vPhi)$ is strictly positive on $D^{N, (p)} \times \mathcal{Q}$.
\end{Lemm}
\begin{proof}
First, consider the case $M\ge 0$, which implies that $p\le N-p$.
Observe that $\mathscr{R}_M(\vPhi)$ can be written as $\{ (\vphi_1 \wedge \cdots \wedge \vphi_p) \otimes (\vphi_1 \wedge \cdots \wedge \vphi_{N - p}) \} \otimes \varOmega$. Consider swapping $\vphi_p$ from the first group $\vphi_1 \wedge \cdots \wedge \vphi_p$ with $\vphi_j$ from the second group $\vphi_1 \wedge \cdots \wedge \vphi_{N - p}$ for $j = 1, \dots, N - p$. The resulting $S_{(p, j)} \mathscr{R}_M(\vPhi)$ vanishes when $j \neq p$, and when $j = p$, we have $S_{(p, p)} \mathscr{R}_M(\vPhi) = \mathscr{R}_M(\vPhi)$. Therefore, by Corollary \ref{IffMM}, it follows that $\vPhi \in \h_M^M$.

Next, consider the case \(M < 0\), that is, when \(p > N - p\). By an analogous argument, we conclude that \(\vPhi \in \h_M^{-M}\).

Furthermore, by choosing $\vphi_1(x) \equiv 1$, $\vphi_2(x) = x$, $\vphi_3(x) = x^2$, up to $\vphi_{N - p}(x) = x^{N - p - 1}$, we can show that the right-hand side of Eq. \eqref{SlVec} simplifies to the product $\prod_{1 \leq i < j \leq p} (x_j - x_i) \times \prod_{p + 1 \leq i < j \leq N} (x_j - x_i)$. Since in the Schrödinger representation, $\varOmega$ is the constant function equal to $1$, it follows from Eq. \eqref{Saseki} that $\mathscr{R}_M(\vPhi)$ is strictly positive on $D^{N, (p)} \times \mathcal{Q}$. 
\end{proof}

\begin{Prop}\label{GSEng3}
Under the assumptions \hyperlink{A1}{\bf (A. 1)}, \hyperlink{A2}{\bf (A. 2)}, and \hyperlink{A3}{\bf (A. 3)}, the equality 
$
E(|M|) = \mathcal{E}(\mathsf{K}_M) = \mathcal{E}(K^{(p)})
$
holds.
\end{Prop}

\begin{proof}
Owing to the positivity-improving assumption \hyperlink{A3}{\bf (A. 3)}, \(K^{(p)}\) possesses a unique ground state, denoted by \(\vPhi_{\mathrm{G}, D^{N,(p)}}\), which is strictly positive. By employing this ground state, we construct the ground state \(\vPhi_{\mathrm{G}}\) of \(\mathsf{K}_M\) via Eqs.  \eqref{PhiDec2} and \eqref{ExtP}. 

From this point, we consider separately the cases \(M \ge 0\) and \(M < 0\). First, in the case \(M \ge 0\), let \(\vPhi \in \h_M\) be any vector satisfying Eq. \eqref{Saseki}. Since \(\langle \vPhi_{\mathrm{G}} \,|\, \vPhi \rangle > 0\), Lemma \ref{SasekiProp} guarantees that \(\vPhi_{\mathrm{G}} \in \h_M^M\). Moreover, Proposition \ref{GSEng2} yields
$
\mathcal{E}\bigl(K^{(p)}\bigr) = \mathcal{E}\bigl(\mathsf{K}_M\bigr).
$
Combining these observations, we deduce that
$
E(M) = \mathcal{E}\bigl(\mathsf{K}_M\bigr) = \mathcal{E}(K^{(p)}).
$
In the case \(M < 0\), an analogous argument shows that \(\vPhi_{\mathrm{G}} \in \h_{M}^{-M}\). Consequently, we obtain
$
E(|M|) = \mathcal{E}\bigl(\mathsf{K}_M\bigr) = \mathcal{E}\bigl(K^{(p)}\bigr).
$
\end{proof}

\subsubsection*{Completion of the Proof of Theorem \ref{AbstEO}}
Fix an arbitrary \(M \in \mathrm{spec}(S^{(3)})\) satisfying \(M \ge 0\).
By Proposition \ref{GSEng2} and  the assumption \hyperlink{A3}{\bf (A. 3)}, \(\mathsf{K}_{M}\) admits a unique ground state \(\vPhi_{{\rm G}, M}\).  Applying \(S^{(-)}\) to \(\vPhi_{{\rm G}, M}\) gives \(S^{(-)} \vPhi_{{\rm G}, M} \in \h_{M - 1}\).  Since \(S^{(-)}\) commutes with \(\mathsf{K}\), we obtain
\[
\mathsf{K}_{M - 1} \, S^{(-)} \vPhi_{{\rm G}, M} \;=\; \mathcal{E}\bigl(\mathsf{K}_M\bigr) \, S^{(-)} \vPhi_{{\rm G}, M},
\]
implying \(\mathcal{E}\bigl(\mathsf{K}_{M - 1}\bigr) \leq \mathcal{E}\bigl(\mathsf{K}_M\bigr)\).  We next show that the equality in this inequality cannot hold.  By Proposition \ref{GSEng2} and  \hyperlink{A3}{\bf (A. 3)}, \(\mathsf{K}_{M-1}\) has a unique ground state, denoted \(\vPhi_{{\rm G}, M-1}\).  It suffices to prove \(\vPhi_{{\rm G}, M-1}\neq S^{(-)} \vPhi_{{\rm G}, M} /\|S^{(-)} \vPhi_{{\rm G}, M}\|\).  By Lemma \ref{RPartGS}, this reduces to showing
\begin{align}
\frac{\mathscr{R}_{M-1}\bigl(\vPhi_{{\rm G}, M-1}\bigr)}
     {\bigl\|\mathscr{R}_{M-1}\bigl(\vPhi_{{\rm G}, M-1}\bigr)\bigr\|}
\;\neq\;
\frac{\mathscr{R}_{M-1}\bigl(S^{(-)} \vPhi_{{\rm G}, M}\bigr)}
     {\bigl\|\mathscr{R}_{M-1}\bigl(S^{(-)} \vPhi_{{\rm G}, M}\bigr)\bigr\|}. \label{NeqG}
\end{align}
Recall that \(\mathscr{R}_{M-1}\bigl(S^{(-)} \vPhi_{{\rm G}, M}\bigr)\) satisfies Eq. \eqref{RS-}.  Moreover, by  \hyperlink{A3}{\bf (A. 3)} and Proposition \ref{GSEng2}, \(\mathscr{R}_{M-1}\bigl(\vPhi_{{\rm G}, M-1}\bigr)\) is strictly positive on \(D^{N, (p+1)}\), whereas the right-hand side of  Eq. \eqref{RS-} shows that \(\mathscr{R}_{M-1}\bigl(S^{(-)} \vPhi_{{\rm G}, M}\bigr)\) does not share this property.  Hence, \eqref{NeqG} follows, establishing that \(\mathcal{E}\bigl(\mathsf{K}_{M - 1}\bigr) < \mathcal{E}\bigl(\mathsf{K}_M\bigr)\).

Combining this with Proposition \ref{GSEng3} then shows that \(E(M - 1) < E(M)\).
\qed

\section{Path Integral Representation of $\ex^{-\beta H_{\rm F} \restriction \mathfrak{K}_{D^{N, (p)}}}$}\label{FKIF}

 \subsection{Main Theorem of Section \ref{FKIF}}

Recall that $\mathfrak{K}_{D^{N, (p)}}$ is defined in Eq. \eqref{DefKDNP}. Furthermore, in applying the results and discussions from Sections \ref{PfofThmEnOrd}, we note that the operator $K$ corresponds to $H_{\rm F}$. Unless there is a risk of confusion, we shall continue to identify $H_{\rm F}$ with $H_{\rm F} \otimes \mathbbm{1}$, as before.

By Proposition \ref{GSEng3}, we have:

\begin{Lemm}
For every $M\in \{-N/2, -N/2+1, \dots, N/2\}$, one has
$E(|M|) = \mathcal{E}(H_{\rm F} \restriction \mathfrak{K}_{D^{N, (p)}})$.
\end{Lemm}

Given the unique properties of this one-dimensional system, to investigate the characteristics of $E(|M|)$, it suffices to analyze $H_{\rm F}$ restricted to $\mathfrak{K}_{D^{N, (p)}}$. Notably,  the analysis of $H_{\rm F} \restriction \mathfrak{K}_{D^{N, (p)}}$ does not require consideration of Fermi statistics. This simplification allows us to examine the ground state energy using a functional integral representation.

In the remainder of this paper, we shall denote the restrictions of $H_{\vepsilon}$ and $H_{\rm F}$ to $\mathfrak{K}_{D^{N, (p)}}$ by the same symbols.

For $\vepsilon > 0$, the following holds:
\begin{align}
\left\langle f \otimes \vOm\,  \Big|\,  \ex^{-\beta H_{\vepsilon}} g \otimes \vOm \right\rangle
= \int_{D^{N, (p)}} d\bs{x} \, \Ex_{\bs{x}} \left[ \mathbbm{1}_{\left\{\tau_p>\beta \right\}}   \,
f(\bs{x}_0)^* g(\bs{x}_{\beta}) \, \ex^{S_{\vepsilon}} \right] \quad (f, g \in L^2(D^{N, (p)})). \label{FKN1}
\end{align}
Here, the action $S_{\vepsilon} \deq S_{\rm el} + S_{{\rm eff}, \vepsilon}$ is defined as follows:
\begin{align}
S_{\rm el} &\deq -\int_0^{\beta} ds\, U(\bs{x}_s), \quad
S_{{\rm eff}, \vepsilon} \deq  \int_0^{\beta} \int_0^{\beta} ds\, dt\, \ex^{-|s - t|} W_{{\rm eff}, \vepsilon}(\bs{x}_s - \bs{x}_t), \label{Act}
\end{align}
with
\be
W_{{\rm eff}, \vepsilon}(\bs{x}) \deq 
\sum_{i, j = 1}^N w(x_i - x_j), \quad
w(x) \deq  \frac{g_L}{2} \sum_{k \in D^*} \ex^{-2\vepsilon k^2} \ex^{\im k x},
\quad g_L \deq  \frac{\sqrt{2}}{L} \alpha;  \label{W}\notag
\ee
In this expression, $(\bs{x}_{s})_s \in C_{\bs{x}}([0, \beta]; \BbbR^N)$ represents Brownian motion starting at $\bs{x}$ at time $s = 0$, and $\Ex_{\bs{x}}$ denotes the expectation with respect to the associated Wiener measure (see, for instance, \cite{Lrinczi2020} for details). The space $C_{\bs{x}}([0, \beta]; \BbbR^N)$ consists of $\BbbR^N$-valued continuous paths $(\bs{x}_{s})_s$ satisfying $\bs{x}_0 = \bs{x}$. Moreover,
the first entry time of $(\bs{x}_s)_s$ into 
${\mathbb R}^N\setminus D^{N, (p)}$ is denoted by 
$\tau_p\left((\bs{x}_s)_s\right)\deq \inf \{s\geq0\, :\,  \bs{x}_s\in {\mathbb R}^N\setminus D^{N, (p)}\}$.

Let us consider the limit $\vepsilon \to +0$ in Eq. \eqref{FKN1}. When $L$ is sufficiently large and using the continuum approximation, we have $\lim_{\vepsilon \to +0} W_{{\rm eff}, \vepsilon}(\bs{x}) \approx \delta(\bs{x})$, where $\delta(\bs{x})$ denotes the Dirac delta function. Substituting this into Eq. \eqref{FKN1}, it becomes unclear whether the integral on the right-hand side is mathematically well-defined. Therefore, a meticulous examination is required to determine whether $S_{{\rm eff}, \vepsilon}$ converges as $\vepsilon \to +0$.

\begin{Thm}\label{PathIntF}
There exists a random variable $S_{{\rm eff}, 0}$ such that the following statements hold:
\begin{itemize}
\item[\rm (i)] For every path $(\bs{x}_{s})_s \in C_{\bs{x}}([0, \beta]; \BbbR^N)$, we have $S_{{\rm eff}, 0} \geq 0$.
\item[\rm (ii)] For every path $(\bs{x}_{s})_s \in C_{\bs{x}}([0, \beta]; \BbbR^N)$, the functional $S_{{\rm eff}, 0}$ is monotonically increasing with respect to $\alpha$.
\item[\rm (iii)] Let $S_0 \deq  S_{\rm el} + S_{{\rm eff}, 0}$. Then,
\be
\lim_{\vepsilon \to +0} \Ex_{\bs{x}} \left[ \mathbbm{1}_{\left\{\tau_p>\beta \right\}}  \ex^{S_{\vepsilon}} \right] = \Ex_{\bs{x}} \left[ \mathbbm{1}_{\left\{\tau_p>\beta \right\}} \ex^{S_0} \right]. \notag
\ee
\item[\rm (iv)] For all $f, g \in L^2(D^{N,(p)})$, we have
\begin{align}
\left\langle f \otimes \vOm\,  \Big|\,  \ex^{-\beta H_{\rm F}} g \otimes \vOm \right\rangle
= \int_{D^{N,(p)}} d\bs{x} \, \Ex_{\bs{x}} \left[ \mathbbm{1}_{\left\{\tau_p>\beta \right\}}   f(\bs{x}_0)^* g(\bs{x}_\beta) \, \ex^{S_0} \right]. \label{FKN2}
\end{align}
\end{itemize}
\end{Thm}

We shall prove this theorem by applying the method developed in  \cite{Gubinelli2014}. Notably, unlike in the Nelson model, energy renormalization is unnecessary here.

\begin{Rem}
\begin{itemize}
\item In \cite{Gubinelli2014}, it was essential to decompose the action $S_{\rm eff, \vepsilon}$ into contributions from the diagonal region (near $s = t$) and the off-diagonal region in the double time integral, analyzing each separately. However, in the Fröhlich model, such a decomposition is unnecessary.
\item Additionally, in \cite{Gubinelli2014}, the explicit forms of the Fourier transforms of certain functions played a significant role. Since we are considering a finite-volume system, we need to replace these parts of the proof in \cite{Gubinelli2014} with an analysis involving Fourier series.
\end{itemize}
\end{Rem}

\subsection{Analysis of $ S_{{\rm eff}, \varepsilon} $}

As the subsequent analysis is independent of the coupling constant $\alpha$, we set $\alpha = 1$ for simplicity. For $ t \in [0, \beta] $, define
$
\xi(t) \deq  \mathrm{e}^{-t}.
$
Following the approach of \cite{Gubinelli2014}, we define
\[
\varphi_{\varepsilon}(x, t) 
\deq \left[
\frac{g_L}{2} \sum_{k \in D^*} \frac{\mathrm{e}^{-2\varepsilon k^2} \mathrm{e}^{\mathrm{i} k x}}{1 + \frac{1}{2} k^2}
\right] \xi(t).
\]
For each fixed \(t\), \(\varphi_{\varepsilon}(x,t)\) is a \(2L\)-periodic function on \(\mathbb{R}\).
Applying Itō's formula, we obtain
\begin{align*}
\varphi_{\varepsilon}\left(x_{i, \beta} - x_{j, s}, \beta - s\right) - \varphi_{\varepsilon}\left(x_{i, s} - x_{j, s}, 0\right) ={} & \int_s^{\beta} \left( \frac{\partial}{\partial x} \varphi_{\varepsilon} \right)\left(x_{i, t} - x_{j, s}, t - s\right) dx_{i, t} \no
& + \int_s^{\beta} \left( \frac{\partial}{\partial t} + \frac{1}{2} \frac{\partial^2}{\partial x^2}  \right)\varphi_{\varepsilon}\left(x_{i, t} - x_{j, s}, t - s\right) dt.
\end{align*}
Let $ w(x) $ be the function defined in Eq.  \eqref{W}. It is readily verified that
\[
\left( \frac{\partial}{\partial t} + \frac{1}{2} \frac{\partial^2}{\partial x^2} \right) \varphi_{\varepsilon}(x, t) = -w(x) \xi(t).
\]
Thus,
\begin{align*}
\int_s^{\beta} \xi(s - t) w\left(x_{i, t} - x_{j, s}\right) dt ={} & \varphi_{\varepsilon}\left(x_{i, s} - x_{j, s}, 0\right) - \varphi_{\varepsilon}\left(x_{i, \beta} - x_{j, s}, \beta - s\right) \no
& + \int_s^{\beta} \left( \frac{\partial}{\partial x} \varphi_{\varepsilon} \right)\left(x_{i, t} - x_{j, s}, t - s\right) dx_{i, t}.
\end{align*}
Substituting this into the definition of $ S_{{\rm eff}, \varepsilon} $ in Eq.  \eqref{Act}, we further decompose it as:
\[
S_{{\rm eff}, \varepsilon} = 2 \beta N \varphi_{\varepsilon}(0, 0) + X_{\varepsilon} + Y_{\varepsilon} + Z_{\varepsilon},
\]
where
\begin{align*}
X_{\varepsilon} & \deq  2 \sum_{\substack{i, j  = 1 \\ i \neq j}}^N \int_0^{\beta} \varphi_{\varepsilon}\left(x_{i, s} - x_{j, s}, 0\right) ds, \\
Y_{\varepsilon} & \deq  2 \sum_{i, j  =  1}^N \int_0^{\beta} ds \int_s^{\beta} \left( \frac{\partial}{\partial x} \varphi_{\varepsilon} \right)\left(x_{i, t} - x_{j, s}, t - s\right) dx_{i, t}, \\
Z_{\varepsilon} & \deq  -2 \sum_{i, j  =  1}^N \int_0^{\beta} \varphi_{\varepsilon}\left(x_{i, \beta} - x_{j, s}, \beta - s\right) ds.
\end{align*}

To analyze the behavior of $ S_{{\rm eff}, \varepsilon} $ as $ \varepsilon \to +0 $, it suffices to examine the asymptotic behavior of $ X_{\varepsilon} $, $ Y_{\varepsilon} $, and $ Z_{\varepsilon} $. For this purpose, we first introduce:
\begin{align*}
g(x) & \deq  \sum_{n \in \mathbb{Z}} \mathrm{e}^{- \sqrt{2} | x + n L |}, \\
K_{\varepsilon}(x) & \deq  \sum_{n \in \mathbb{Z}} \mathrm{e}^{- (x + 2 n L)^2 / (8 \varepsilon)}, \\
h(x) & \deq  \sum_{n \in \mathbb{Z}} \vartheta(x + n L) \, \mathrm{e}^{- \sqrt{2} | x + n L |},
\end{align*}
where
\[
\vartheta(x)  \deq 
\begin{cases}
1, & x > 0, \\
-1, & x < 0.
\end{cases}
\]
The function \(h(x)\) is defined for \(x \notin L\BbbZ\); for convenience, we set \(h(x)=0\) for \(x\in L\BbbZ\).
It is evident that these functions are $ 2L $-periodic. By expanding them into Fourier series, we obtain:
\begin{align*}
g(x) &= \frac{\sqrt{2}}{L} \sum_{k \in D^*} \frac{1}{1 + \frac{1}{2} k^2} \mathrm{e}^{\mathrm{i} k x}, \\
K_{\varepsilon}(x) &= \frac{\sqrt{2\pi \varepsilon}}{L} \sum_{k \in D^*} \mathrm{e}^{-2\varepsilon k^2} \mathrm{e}^{\mathrm{i} k x}, \\
h(x) &= -\frac{1}{L} \sum_{k \in D^*} \frac{\mathrm{i} k}{1 + \frac{1}{2} k^2} \mathrm{e}^{\mathrm{i} k x}.
\end{align*}
Here, we set \(h(x)=0\) for \(x \in L\BbbZ\).\footnote{This may be interpreted as choosing a specific summation convention for the series:
\[
\sum_{k\in D^*}\frac{k\, \ex^{\im k x}}{1+k^2/2}
=\lim_{\kappa\to \infty} \sum_{|k|\le \kappa}\frac{k\, \ex^{\im kx}}{1+k^2/2}.
\]}
Recall that the convolution is defined by
\[
(K_{\varepsilon} * f)(x) = \int_{-L}^{L} K_{\varepsilon}(y) f(x - y) dy.
\]
From this, we deduce:
\begin{align}
\varphi_{\varepsilon}(x, t) &= \frac{1}{4 \sqrt{2\pi \varepsilon}} (g * K_{\varepsilon})(x) \, \xi(t), \notag\\
\left( \frac{\partial}{\partial x} \varphi_{\varepsilon} \right)(x, t) &= -\frac{1}{4 \sqrt{2\pi \varepsilon}} (h * K_{\varepsilon})(x) \, \xi(t). \label{phie} 
\end{align}
For $ \varepsilon = 0 $, we have
\begin{align}
\varphi_0(x, t) = \frac{1}{2} g(x) \xi(t), \quad
\frac{\partial}{\partial x} \varphi_0(x, t) = -\frac{1}{2} h(x) \xi(t). \label{phi0}
\end{align}

\begin{Lemm}\label{Lemmphi}
The following hold:
\begin{itemize}
\item[\rm (i)] There exists a constant $ C $, independent of $ x $, $ t $, and $ \varepsilon $, such that
\[
| \varphi_{\varepsilon}(x, t) | \leq C.
\]
\item[\rm (ii)] The convergence
\[
\lim_{\varepsilon \to +0} \varphi_{\varepsilon}(x, t) = \varphi_0(x, t)
\]
is uniform on $ \BbbR \times [0, \beta] $.
\end{itemize}
\end{Lemm}

\begin{proof}
(i) For $x\in D $, $ g(x) $ can be expressed as:
\[
g(x) = \mathrm{e}^{-\sqrt{2} | x |} + \frac{2 \, \mathrm{e}^{-\sqrt{2} L}}{1 - \mathrm{e}^{-\sqrt{2} L}} \cosh(\sqrt{2} x).
\]
Moreover, we observe that
\begin{align}
\frac{1}{2 \sqrt{2\pi \varepsilon}} \int_{-L}^{L} K_{\varepsilon}(x) dx = 1, \quad
\lim_{\varepsilon \to +0} \int_{\delta \leq | x | \leq L} K_{\varepsilon}(x) dx = 0. \label{Ke}
\end{align}
Utilizing the elementary inequalities $ 0 \leq \xi(t) \leq \xi(0)(=1) $ and
\[
| g(x) | \leq C_0\deq  1 + \frac{2 \, \mathrm{e}^{-\sqrt{2} L}}{1 - \mathrm{e}^{-\sqrt{2} L}} \cosh(\sqrt{2} L)  \quad (x \in \overline{D}),
\]
we deduce that
\[
| (g * K_{\varepsilon})(x) | \leq C_0.
\]
Therefore, setting $ C= \xi(0) C_0 / 2 $, we obtain the desired inequality.

(ii) From \eqref{Ke}, it follows that $ K_{\varepsilon} $ acts as an approximate identity, see, e.g., \cite{stein2011fourier}. Hence, by standard results in analysis, we conclude that the convergence is uniform on $\overline{D}$, yielding the desired result.
\end{proof}

\begin{Lemm}\label{sphiLemm}
The following statements hold:
\begin{itemize}
\item[\rm (i)] There exists a constant $C$, independent of $x$, $t$, and $\vepsilon$, such that
\be
\left|\frac{\partial}{\partial x}\varphi_{\vepsilon}(x, t)\right| \leq C. \notag
\ee
\item[\rm (ii)] For $x \in \BbbR$, define
\be
\delta_{\vepsilon}(x) \deq  \left| \frac{1}{2\sqrt{2\pi \vepsilon}} (h * K_{\vepsilon})(x) - h(x) \right|. \notag
\ee
Then we have $\displaystyle \lim_{\vepsilon \to +0} \delta_{\vepsilon}(x) = 0$ a.e. $x$ and, for any $x \in \BbbR$, 
\be
\delta_{\vepsilon}(x) \leq 2 + \frac{4 \ex^{-\sqrt{2} L}}{1 - \ex^{-\sqrt{2} L}} \sinh(\sqrt{2} L). \label{DeI}
\ee
Furthermore,
\be
\left| \frac{\partial}{\partial x} \varphi_{\vepsilon}(x, t) - \frac{\partial}{\partial x} \varphi_{0}(x, t) \right| \leq \frac{\delta_{\vepsilon}(x)}{2}. \label{DifDphi}
\ee
\end{itemize}
\end{Lemm}

\begin{proof}
(i) For $x\in D$, we can express
\be
h(x) = \vartheta(x) \, \ex^{-\sqrt{2} | x |} + \frac{2 \ex^{-\sqrt{2} L}}{1 - \ex^{-\sqrt{2} L}} \sinh(\sqrt{2} x). \notag
\ee
This yields
\be
| h(x) | \leq C_0\deq 1 + \frac{2 \ex^{-\sqrt{2} L}}{1 - \ex^{-\sqrt{2} L}} \sinh(\sqrt{2} L). \label{hEst}
\ee
From \eqref{Ke}, it follows that $| (h * K_{\vepsilon})(x) | \leq C_0$. Setting $C = \xi(0) C_0 / 2$, we obtain the desired result.

(ii) The inequality \eqref{DeI} follows immediately from \eqref{hEst}. From Eqs. \eqref{phie} and \eqref{phi0}, we have
\be
\text{LHS of \eqref{DifDphi}} = | \xi(t) | \frac{1}{2} \delta_{\vepsilon}(x) \leq \xi(0) \frac{1}{2} \delta_{\vepsilon}(x). \notag
\ee
\end{proof}

\begin{Lemm}\label{RLemm}
For all $\beta > 0$ and any path $(\boldsymbol{x}_s)_s$, the following hold:
\begin{itemize}
\item[\rm (i)] $\displaystyle \lim_{\vepsilon \to +0} \varphi_{\vepsilon}(0, 0) = \varphi_{0}(0, 0)$.
\item[\rm (ii)] $| X_{\vepsilon} | \leq 2 N^2 \beta C$ and
$\displaystyle 
\lim_{\vepsilon \to +0} X_{\vepsilon} = X_{0}.
$
\item[\rm (iii)] $| Z_{\vepsilon} | \leq 2 N^2 \beta C$ and
$\displaystyle 
\lim_{\vepsilon \to +0} Z_{\vepsilon} = Z_{0}.
$
\end{itemize}
The constant \( C \) appearing in {\rm (ii)} and {\rm (iii)} is independent of \( N, \beta, \vepsilon, \bs{x} \), and the path \((\boldsymbol{x}_s)_s\).
\end{Lemm}

\begin{proof}
These results follow directly from Lemma \ref{Lemmphi} and Lebesgue's dominated convergence theorem.
\end{proof}

Next, we analyze $Y_{\vepsilon}$, which is more challenging than analyzing $X_{\vepsilon}$ and $Z_{\vepsilon}$. By exchanging the order of integration over $s$ and $t$, $Y_{\vepsilon}$ can be written as:
\be
Y_{\vepsilon} = \sum_{i=1}^N \int_0^{\beta} \vPhi_{\vepsilon, t}^{(i)} \, d x_{i, t} = \int_0^{\beta} \bs{\vPhi}_{\vepsilon, t} \cdot d \boldsymbol{x}_t, \notag
\ee
where
\be
\bs{\vPhi}_{\vepsilon, t} \deq  \left( \vPhi_{\vepsilon, t}^{(1)}, \dots, \vPhi_{\vepsilon, t}^{(N)} \right), \quad \vPhi_{\vepsilon, t}^{(i)} \deq  2 \sum_{j=1}^N \int_0^t \left( \frac{\partial}{\partial x} \varphi_{\vepsilon} \right) \left( x_{i, t} - x_{j, s}, t - s \right) \, d s. \notag
\ee
Note that when \(\vepsilon=0\), although \(\vphi_0(x, s)\) is discontinuous, the integration with respect to \(s\) in the above expression is well-defined when interpreted as a Lebesgue integral.
Letting $\delta Y_{\vepsilon} \deq Y_{\vepsilon} - Y_0$, we have
\begin{align*}
\delta Y_{\vepsilon} = \int_0^{\beta} \delta \bs{\vPhi}_{\vepsilon, t} \cdot d \boldsymbol{x}_t,
\end{align*}
where
$
\delta \bs{\vPhi}_{\vepsilon, t} \deq  \bs{\vPhi}_{\vepsilon, t} - \bs{\vPhi}_{0, t}.
$

\begin{Lemm}\label{PhiLemm}
For all $\beta > 0$, $\vepsilon>0$ and path $(\boldsymbol{x}_s)_s$, the following hold:
\begin{itemize}
\item[\rm (i)]
$\displaystyle 
\int_0^{\beta} \left| \bs{\vPhi}_{\vepsilon, t} \right|^2 \, d t \leq 4 C^2 \beta^3 N^3,\notag
$
where
$
\left| \bs{\vPhi}_{\vepsilon, t} \right| \deq \left\{ \sum_{i=1}^N \left( \varPhi_{\vepsilon, t}^{(i)} \right)^2 \right\}^{1/2}.
$ Moreover, the constant $C$
is independent of    \( N, \beta, \vepsilon, \bs{x} \), and the path \((\boldsymbol{x}_s)_s\).
\item[\rm (ii)]
$\displaystyle 
\lim_{\vepsilon \to +0} \int_0^{\beta} \left| \delta \bs{\vPhi}_{\vepsilon, t} \right|^2 \, d t = 0. \notag
$
\end{itemize}
\end{Lemm}

\begin{proof}
(i) By Lemma \ref{sphiLemm},
\begin{align*}
\int_0^{\beta} \left| \varPhi_{\vepsilon, t}^{(i)} \right|^2 \, d t &\leq 4 \int_0^{\beta} \left[ \sum_{j=1}^N \int_0^t \left| \left( \frac{\partial}{\partial x} \varphi_{\vepsilon} \right) \left( x_{i, t} - x_{j, s}, t - s \right) \right| \, d s \right]^2 \, d t \\
&\leq 4 C^2 \beta^3 N^2.
\end{align*}

(ii) From Lemma \ref{sphiLemm} (ii) and Lebesgue's dominated convergence theorem, we have
\begin{align*}
&\int_0^{\beta} \left| \delta \bs{\vPhi}_{\vepsilon, t} \right|^2 \, d t\no
 &\leq 4 \int_0^{\beta} \left[ \sum_{j=1}^N \int_0^t \left| \left( \frac{\partial}{\partial x} \varphi_{\vepsilon} \right) \left( x_{i, t} - x_{j, s}, t - s \right) - \left( \frac{\partial}{\partial x} \varphi_{0} \right) \left( x_{i, t} - x_{j, s}, t - s \right) \right| \, d s \right]^2 \, d t \\
&\leq \sum_{i=1}^N \int_0^{\beta} \left[ \sum_{j=1}^N \int_0^t  \delta_{\vepsilon}(x_{i, t} - x_{j, s}) \, d s \right]^2 \, d t \\
&\to 0 \quad (\vepsilon \to +0).
\end{align*}
\end{proof}

\begin{Lemm}\label{YLemm}
For any $a\in \mathbb{R}$, the following hold:
\begin{itemize}
\item[\rm (i)]
$\displaystyle 
\Ex_{\bs{x}}\left[
\ex^{a Y_{\vepsilon}}
\right]
\le \left(\Ex_{\bs{x}}\left[\ex^{a^2\int_0^{\beta} |\bs{\vPhi}_{\vepsilon, t}|^2 dt} \right]
\right)^{1/2}.
$

\item[\rm (ii)]
$\displaystyle 
\Ex_{\bs{x}}\left[
\ex^{a \delta Y_{\vepsilon}}
\right]
\le \left(\Ex_{\bs{x}}\left[\ex^{a^2 \int_0^{\beta} |\delta\bs{\vPhi}_{\vepsilon, t}|^2 dt} \right]
\right)^{1/2}.
$
\item[\rm (iii)]
$\displaystyle 
\lim_{\vepsilon\to +0}\Ex_{\bs{x}}\left[\ex^{a \delta Y_{\vepsilon}}\right]=1.
$
\end{itemize}
\end{Lemm}
\begin{proof}
(i) For simplicity of notation, we consider the case $a=1$ only. By applying  Girsanov's theorem \cite[Proposition 2.171]{Lrinczi2020},
\begin{align}
\Ex_{\bs{x}}\left[
\exp\left\{
2\int_0^{\beta}\bs{\vPhi}_{\vepsilon, t} \cdot d\bs{x}_t-\frac{1}{2}2^2 \int_0^{\beta} |\bs{\vPhi}_{\vepsilon, t}|^2dt
\right\}
\right]=1. \notag
\end{align}
Thus,  
\begin{align*}
\left(
\Ex_{\bs{x}}\left[
\ex^{Y_{\vepsilon}}
\right]
\right)^2&=\left(
\Ex_{\bs{x}}\left[
\ex^{\int_0^{\beta}\bs{\vPhi}_{\vepsilon, t}\cdot d\bs{x}_t-\int_0^{\beta} |\bs{\vPhi}_{\vepsilon, t}|^2dt}\times
\ex^{\int_0^{\beta} |\bs{\vPhi}_{\vepsilon, t}|^2 dt}
\right]
\right)^2\no
&\le 
\Ex_{\bs{x}}\left[
\ex^{2\int_0^{\beta}\bs{\vPhi}_{\vepsilon, t}\cdot  d\bs{x}_t-2 \int_0^{\beta} |\bs{\vPhi}_{\vepsilon, t}|^2dt}
\right]\times \Ex_{\bs{x}}\left[\ex^{\int_0^{\beta} |\bs{\vPhi}_{\vepsilon, t}|^2 dt} \right]\no
&= \Ex_{\bs{x}}\left[\ex^{\int_0^{\beta} |\bs{\vPhi}_{\vepsilon, t}|^2 dt} \right].
\end{align*}
Here, the first inequality uses the Cauchy--Schwarz inequality. The proof of (ii) is similar.

(iii) Again, by Girsanov's theorem,
\begin{align}
\Ex_{\bs{x}}\left[
\exp\left\{
\int_0^{\beta}\delta \bs{\vPhi}_{\vepsilon, t}\cdot d\bs{x}_t-\frac{1}{2} \int_0^{\beta} |\delta\bs{\vPhi}_{\vepsilon, t}|^2dt
\right\}
\right]=1. \label{Gir2}
\end{align}
Thus,
\begin{align}
\left(
\Ex_{\bs{x}}\left[\ex^{\delta Y_{\vepsilon}}\right]-1
\right)^2
&=
\left(
\Ex_{\bs{x}}\left[\ex^{\int_0^{\beta} \delta \bs{\vPhi}_{\vepsilon, t}\cdot d\bs{x}_t}-1\right]
\right)^2\no
&= \left\{
\Ex_{\bs{x}}\left[
\ex^{\int_0^{\beta} \delta \bs{\vPhi}_{\vepsilon, t}\cdot d\bs{x}_t}
\left(
1-\ex^{-\frac{1}{2} \int_0^{\beta} |\delta \bs{\vPhi}_{\vepsilon, t}|^2 d t}
\right)
\right]
\right\}^2\no
&\le 
\Ex_{\bs{x}}\left[
\ex^{2\int_0^{\beta} \delta \bs{\vPhi}_{\vepsilon, t}\cdot d\bs{x}_t}
\right]
\Ex_{\bs{x}}\left[
\left(
1-\ex^{-\frac{1}{2} \int_0^{\beta} |\delta \bs{\vPhi}_{\vepsilon, t}|^2 d t}
\right)^2
\right]\no
&\le \left(
\Ex_{\bs{x}}\left[
\ex^{8\int_0^{\beta} |\delta \bs{\vPhi}_{\vepsilon, t}|^2dt}
\right]
\right)^{1/2}
\Ex_{\bs{x}}\left[\left(
\frac{1}{2} \int_0^{\beta} |\delta \bs{\vPhi}_{\vepsilon, t}|^2 d t
\right)^2
\right]\no
&\to 0\quad \text{as $\vepsilon\to +0$}.\notag
\end{align}
Here, in the second equality and the second inequality, we used Eq. \eqref{Gir2}. The convergence of the limit follows from the dominated convergence theorem.
\end{proof}

\begin{Prop}\label{S-Prop}
For any $a\in \mathbb{R}$, the following holds:
\begin{itemize}
\item[\rm (i)] There exists a constant $C$ independent of $\vepsilon$ and $\bs{x}$ such that
$\displaystyle 
\Ex_{\bs{x}}\left[\ex^{a S_{{\rm eff}, \vepsilon}} \right]\le C. \notag
$
\item[\rm (ii)] Let $\delta S_{{\rm eff}, \vepsilon} \deq S_{{\rm eff}, \vepsilon}-S_{\rm eff, 0}$. Then,
$\displaystyle 
\lim_{\vepsilon\to +0}\Ex_{\bs{x}}\left[\ex^{a \delta S_{{\rm eff}, \vepsilon}}\right]=1.\notag
$
\end{itemize}
\end{Prop}
\begin{proof}
(i) This follows from Lemmas \ref{RLemm}, \ref{PhiLemm}, and \ref{YLemm}.

(ii) First, consider the following decomposition:
\[
1-\Ex_{\bs{x}}\left[
\ex^{a \delta S_{{\rm eff}, \vepsilon} }
\right]
=I_1+I_2,
\]
where
\begin{align*}
I_1&\deq 1-\Ex_{\bs{x}}\left[
\ex^{a\delta \vphi_{\vepsilon}(0, 0)+a \delta X_{\vepsilon}+a \delta Z_{\vepsilon}}
\right],\quad 
I_2\deq 
\Ex_{\bs{x}}\left[
\ex^{a\delta \vphi_{\vepsilon}(0, 0)+a \delta X_{\vepsilon}+a \delta Z_{\vepsilon}}
\right]-
\Ex_{\bs{x}}\left[
\ex^{a \delta S_{{\rm eff}, \vepsilon} }
\right].
\end{align*}
Here, $\delta X_{\vepsilon}\deq X_{\vepsilon}-X_0$, $\delta Z_{\vepsilon}\deq Z_{\vepsilon}-Z_0$, and $\delta \vphi_{\vepsilon}(0, 0)\deq \vphi_{\vepsilon}(0, 0)-\vphi_{0}(0, 0)$. From Lemma \ref{RLemm}, it is easy to see that $I_1\to 0$ as $\vepsilon\to +0$. Now, consider $I_2$. By Lemmas \ref{RLemm} and \ref{YLemm},
\begin{align*}
I_2^2&=\left\{
\Ex_{\bs{x}}\left[
\ex^{a\delta \vphi_{\vepsilon}(0, 0)+a \delta X_{\vepsilon}+a \delta Z_{\vepsilon}}
\left(
1-\ex^{a \delta Y_{\vepsilon}}
\right)
\right]
\right\}^2\no
&\le \Ex_{\bs{x}}\left[
\ex^{2a\delta \vphi_{\vepsilon}(0, 0)+2a \delta X_{\vepsilon}+2a \delta Z_{\vepsilon}}
\right]\times  \Ex_{\bs{x}}\left[
\left(
1-\ex^{a \delta Y_{\vepsilon}}
\right)^2
\right]\to 0\ (\vepsilon\to +0).
\end{align*}
This completes the proof of Proposition \ref{S-Prop}.
\end{proof}

\subsection{Proof of Theorem \ref{PathIntF}}
\subsubsection*{Proof of (i), (ii) and (iii)}
By the Cauchy--Schwarz inequality, we have
\begin{align*}
&\left| \Ex_{\bs{x}}\left[\mathbbm{1}_{\left\{\tau_p>\beta \right\}}  \ex^{S_{\vepsilon}}\right]-\Ex_{\bs{x}}\left[\mathbbm{1}_{\left\{\tau_p>\beta \right\}}  \ex^{S_{0}}\right] \right|\no
= &  \left| \Ex_{\bs{x}}\left[\mathbbm{1}_{\left\{\tau_p>\beta \right\}}  \ex^{S_{\rm el}}\left(\ex^{ S_{{\rm eff}, \vepsilon}}-\ex^{S_{\rm eff, 0}}\right)\right] \right|\no
\le & \left\{\Ex_{\bs{x}}\left[\mathbbm{1}_{\left\{\tau_p>\beta \right\}}  \ex^{2S_{\rm el}}\right]\right\}^{1/2}
\left\{\Ex_{\bs{x}}\left[\left(\ex^{ S_{{\rm eff}, \vepsilon}}-\ex^{S_{\rm eff, 0}}\right)^2\right]\right\}^{1/2}.\label{DiffES}
\end{align*}
Here,
\[
\Ex_{\bs{x}}\left[\mathbbm{1}_{\{\tau_p>\beta\}} \ex^{2S_{\rm el}}\right]
=\left(\ex^{-\beta\,(T_{D^N}+2U)}\,\omega\right)(\bs{x}),
\]
where $T_{D^N}$ in  the Dirichlet Laplacian on  $D^{N}$ and \(\omega\) is the function on \(D^N\) that is identically equal to 1. Thus, \(\Ex_{\bs{x}}\left[\mathbbm{1}_{\{\tau_p>\beta\}} \ex^{2S_{\rm el}}\right]\) is finite for  almost every $\bs{x}$.
On the other hand, by the Cauchy--Schwarz inequality and Proposition \ref{S-Prop},
\begin{align}
\Ex_{\bs{x}}\left[\left(\ex^{ S_{{\rm eff}, \vepsilon}}-\ex^{S_{\rm eff, 0}}\right)^2\right]
&=\Ex_{\bs{x}}\left[\ex^{2S_{\rm eff, \vepsilon}}\left(1-\ex^{-\delta S_{{\rm eff}, \vepsilon}}\right)^2\right]\no
&\le C \left\{ \Ex_{\bs{x}}\left[\left(1-\ex^{-\delta S_{{\rm eff}, \vepsilon}}\right)^4\right]\right\}^{1/2}
\to 0\ (\text{as $\vepsilon \to +0$}).\label{DConv0}
\end{align}
Thus, (iii) is proven.

Since $\Ex_{\bs{x}}\left[\left(\ex^{ S_{{\rm eff}, \vepsilon}} - \ex^{S_{\rm eff, 0}}\right)^2\right] \to 0$ as $\vepsilon \to +0$, it follows that a subnet of $(S_{{\rm eff}, \vepsilon})_{\vepsilon > 0}$ converges to $S_{\rm eff, 0}$ almost surely.
 For any $\vepsilon>0$ and path $(\bs{x}_{s})_s$, we have $S_{\rm eff, \vepsilon}>0$, hence $S_{\rm eff, 0}\ge 0$. Moreover, since $S_{\rm eff, \vepsilon}$ is monotonically increasing with respect to $\alpha$, $S_{\rm eff, 0}$ is also monotonic. Therefore, (i) and (ii) are established. \qed

\subsubsection*{Proof of (iv)}
\textbf{Step 1.}
First, we prove the claim for $U\in L^{\infty}$. According to Nelson \cite{Nelson1964}, $H_{\vepsilon}$ converges strongly in the resolvent sense to $H_{\rm F}$ as $\vepsilon\to +0$. Thus,
\[\left\la f \otimes \vOm\, \Big|\, \ex^{-\beta H_{\vepsilon}} g\otimes \vOm\right\ra\to \left\la f \otimes \vOm\, \Big|\, \ex^{-\beta H_{\rm F}} g\otimes \vOm\right\ra \quad (\vepsilon\to +0).
\]
Next, we show that the right-hand side of Eq. \eqref{FKN1} converges to the right-hand side of Eq.  \eqref{FKN2}. By the Cauchy--Schwarz inequality, we obtain:
\begin{align}
\left|\int_{D^{N, (p)}} d\bs{x} \Ex_{\bs{x}}\left[\mathbbm{1}_{\left\{\tau_p>\beta \right\}}  
f(\bs{x}_0)^* g(\bs{x}_{\beta})\,  \ex^{S_{\vepsilon}}
\right]-\int_{D^{N, (p)}} d\bs{x} \Ex_{\bs{x}}\left[\mathbbm{1}_{\left\{\tau_p>\beta \right\}}  
f(\bs{x}_0)^* g(\bs{x}_{\beta})\,  \ex^{S_0}
\right]\right|\no
\le \int_{D^{N, (p)}} d\bs{x} f(\bs{x})^*
\left\{
\Ex_{\bs{x}}\left[\left(\ex^{ S_{{\rm eff}, \vepsilon}}-\ex^{S_{\rm eff, 0}}\right)^2\right]
\right\}^{1/2}
\left\{
\Ex_{\bs{x}}\left[
\mathbbm{1}_{\left\{\tau_p>\beta \right\}}  \ex^{2S_{\rm el}}|g(\bs{x}_{\beta})|^2
\right]
\right\}^{1/2}. \label{InqEx1}
\end{align}
By Proposition \ref{S-Prop}, $\Ex_{\bs{x}}\left[\left(\ex^{ S_{{\rm eff}, \vepsilon}}-\ex^{S_{\rm eff, 0}}\right)^2\right]$ is uniformly bounded with respect to $\bs{x}$ and $\vepsilon$. On the other hand, as shown below, $\left\{\Ex_{\bs{x}}\left[
\mathbbm{1}_{\left\{\tau_p>\beta \right\}}  \ex^{2S_{\rm el}}|g(\bs{x}_{\beta})|^2
\right]
\right\}^{1/2}\in L^2(D^{N, (p)})$. Hence, by the dominated convergence theorem,
the right hand side  of \eqref{InqEx1} converges to zero as $\vepsilon\to +0$.

We show that $\left\{\Ex_{\bs{x}}\left[
\mathbbm{1}_{\left\{\tau_p>\beta \right\}}  \ex^{2S_{\rm el}}|g(\bs{x}_{\beta})|^2
\right]
\right\}^{1/2}\in L^2(D^{N, (p)})$. Let $P_{D^{N, (p)}}^{\beta}(\bs{x}, \bs{y})$ be the integral kernel of $\ex^{-\beta T_{D^{N, (p)}}}$, where $T_{D^{N,(p)}}$ in  the Dirichlet Laplacian on the domain $D^{N,(p)}$. Then,
\begin{align*}
\int_{D^{N, (p)}} d\bs{x}\Ex_{\bs{x}}\left[
\mathbbm{1}_{\left\{\tau_p>\beta \right\}}  \ex^{2S_{\rm el}}|g(\bs{x}_{\beta})|^2
\right]
&\le 
\ex^{2\beta \|U\|_{\infty}}
\int_{D^{N, (p)}} d\bs{x}\Ex_{\bs{x}}\left[
\mathbbm{1}_{\left\{\tau_p>\beta \right\}}  |g(\bs{x}_{\beta})|^2
\right]\no
&=\ex^{2\beta \|U\|_{\infty}}\int_{D^{N, (p)}} d\bs{x}\int_{D^{N, (p)}} d\bs{y} P_{D^{N, (p)}}^{\beta}(\bs{x}, \bs{y}) |g(\bs{y})|^2\no
&\le \ex^{2\beta \|U\|_{\infty}}\|g\|^2_2.
\end{align*}
Here, we use $\int_{D^{N, (p)}} d\bs{x} P_{D^{N, (p)}}^{\beta}(\bs{x}, \bs{y})\le 1$ \cite[Theorem 7.13]{grigor2012heat}.

\textbf{Step 2.}
Now, we prove the claim for general $U$ and complete the proof of (iv).
For each $n\in \mathbb{N}$, define $U_n\in L^{\infty}(D^{N, (p)})$ by:
\[U_n(x)=U(x) \quad \text{if } |U(x)| \le n, \quad U_n(x)=n \quad \text{if } |U(x)|\ge n.\]
Let $H_{\rm F}^{(n)}$ be the Fr\"{o}hlich Hamiltonian with $U$ replaced by $U_n$. Then, $H_{\rm F}^{(n)}$ converges strongly in the resolvent sense to $H_{\rm F}$. Moreover, by \textbf{Step 1}, Eq. \eqref{FKN2} holds for each $n\in \mathbb{N}$.
Accordingly, let \(S_0^{(n)}\) and \(S_{\rm el}^{(n)}\) denote the actions defined by replacing \(U\) with \(U_n\) in \(S_0\) and \(S_{\rm el}\), respectively.
 By the Cauchy--Schwarz inequality,
\begin{align*}
&\left|
\int_{D^{N, (p)}} d\bs{x} \Ex_{\bs{x}}\left[\mathbbm{1}_{\left\{\tau_p>\beta \right\}}  
f(\bs{x}_0)^* g(\bs{x}_{\beta})\,  \ex^{S_0^{(n)}}
\right]-\int_{D^{N, (p)}} d\bs{x} \Ex_{\bs{x}}\left[\mathbbm{1}_{\left\{\tau_p>\beta \right\}}  
f(\bs{x}_0)^* g(\bs{x}_{\beta})\,  \ex^{S_0}
\right]
\right|\no
\le & \int_{D^{N, (p)}} d\bs{x} 
|f(\bs{x})|
\Ex_{\bs{x}}\left[\mathbbm{1}_{\left\{\tau_p>\beta \right\}}  
 |g(\bs{x}_{\beta})|\,  \ex^{S_{\rm eff, 0}}
 \left|
  \ex^{S^{(n)}_{\rm el}}-\ex^{S_{\rm el}}
 \right|
\right]\no
\le & \int_{D^{N, (p)}} d\bs{x} |f(\bs{x})|
\left\{
\Ex_{\bs{x}}\left[\mathbbm{1}_{\left\{\tau_p>\beta \right\}}  
 |g(\bs{x}_{\beta})|\,  \ex^{2S_{\rm eff, 0}}\right]
 \right\}^{1/2}
 \left\{
\Ex_{\bs{x}}\left[\mathbbm{1}_{\left\{\tau_p>\beta \right\}}  
 |g(\bs{x}_{\beta})|\,   \left(
  \ex^{S^{(n)}_{\rm el}}-\ex^{S_{\rm el}}
 \right)^2\right]
 \right\}^{1/2}\no
 \le & 
 \left(
  \int_{D^{N, (p)}} d\bs{x} |f(\bs{x})|\Ex_{\bs{x}}\left[\mathbbm{1}_{\left\{\tau_p>\beta \right\}}  
 |g(\bs{x}_{\beta})|\,  \ex^{2S_{\rm eff, 0}}\right]
 \right)^{1/2}\times \no
 &\times
 \left(
  \int_{D^{N, (p)}} d\bs{x} |f(\bs{x})|\Ex_{\bs{x}}\left[\mathbbm{1}_{\left\{\tau_p>\beta \right\}}  
 |g(\bs{x}_{\beta})|\,  \left(
  \ex^{S^{(n)}_{\rm el}}-\ex^{S_{\rm el}}
 \right)^2\right]
 \right)^{1/2}\no
 \deq & I_1\times I_2.
\end{align*}
First, we confirm that $I_1<\infty$. By Proposition \ref{S-Prop},
\begin{align*}
I_1^2
 \le&
 \int_{D^{N, (p)}} d\bs{x} |f(\bs{x})|
 \left\{
 \Ex_{\bs{x}}\left[\mathbbm{1}_{\left\{\tau_p>\beta \right\}}  
 |g(\bs{x}_{\beta})|^2\right]
 \right\}^{1/2}\left\{
 \Ex_{\bs{x}}\left[\mathbbm{1}_{\left\{\tau_p>\beta \right\}}  
 \ex^{4S_{\rm eff, 0}}\right]
 \right\}^{1/2}\no
 \le& C \|f\|\|g\|.
\end{align*}
Next, we show that $I_2\to 0\ (n\to \infty)$. By the Feynman--Kac formula for the Schr\"{o}dinger semigroup,
\begin{align*}
&\Ex_{\bs{x}}\left[\mathbbm{1}_{\left\{\tau_p>\beta \right\}}  
 |g(\bs{x}_{\beta})|\,  \left(
  \ex^{S^{(n)}_{\rm el}}-\ex^{S_{\rm el}}
 \right)^2\right]\no
 =&\left\{
 \ex^{-\beta (T_{D^{N, (p)}}+2U_n)} -2\ex^{-\beta (T_{D^{N, (p)}}+U+U_n)}
 + \ex^{-\beta (T_{D^{N, (p)}}+2U)}
 \right\}|g|(\bs{x})
\end{align*}
It follows that $T_{D^{N, (p)}}+2U_n$ and $T_{D^{N, (p)}}+U+U_n$ converge strongly in the resolvent sense to $T_{D^{N, (p)}}+2U$, hence
\begin{align*}
I_2^2=\left\la |f|\,  \Big|\,  \left\{\ex^{-\beta (T_{D^{N, (p)}}+2U_n)} -2\ex^{-\beta (T_{D^{N, (p)}}+U+U_n)}
 + \ex^{-\beta (T_{D^{N, (p)}}+2U)}
 \right\}|g|\right\ra
 \to 0 \ (n\to \infty).
\end{align*}
Therefore, Eq. \eqref{FKN2} holds for general $U$.
\qed

\section{Proof of Theorem \ref{ExGs}}\label{PfExGs}

\subsection{Strategy of the Proof}

First, we outline the strategy for proving Theorem \ref{ExGs}. Our proof utilizes the Schrödinger representation $L^2(\mathcal{Q})$ of the boson Fock space $\F$. To specify which representation is used, let \( J \) be the unitary operator from \( L^2(\mathcal{Q}) \) to \(\F\). The following theorem is pivotal for the proof of Theorem \ref{ExGs}:

\begin{Thm}\label{Hyp2}
For every \(p\), the operator \(H_{\rm F} \restriction \mathfrak{K}_D^{N,(p)}\) possesses a unique, strictly positive ground state.
\end{Thm}

The proof of this theorem is presented in Subsections \ref{PfHyC} and \ref{Unique}. The principal idea is as follows: Because the analysis is conducted within the bounded domain \( D^{N, (p)} \times \mathcal{Q} \), the existence of a ground state for \( H_{\rm F} \restriction \mathfrak{K}_{D^{N, (p)}} \) is established by demonstrating that the semigroup generated by \( H_{\rm F} \restriction \mathfrak{K}_{D^{N, (p)}} \) is hypercontractive \cite{Nelson1973}. Furthermore, we demonstrate that the heat semigroup generated by \( H_{\rm F} \restriction \mathfrak{K}_{D^{N, (p)}} \) is positivity-improving in the Schrödinger representation. This property, combined with the Perron--Frobenius theorem, guarantees the uniqueness and strict positivity of the ground state.

\subsubsection*{Proof of Theorem \ref{ExGs}, given Theorem \ref{Hyp2}}
By Proposition \ref{GSEng2}, it suffices to show that \( H_{\rm F} \restriction \mathfrak{K}_{D^{N, (p)}} \) has a unique ground state for each \( p \in \{1, \dots, N\} \). This follows immediately from Theorem \ref{Hyp2}.\qed

\subsection{Preliminaries}

We use the approach of \cite{Hiroshima2021, MATTE2018} to prove Theorem \ref{Hyp2}. In the analysis of the renormalized Nelson model in \cite{Hiroshima2021, MATTE2018}, a particular Feynman--Kac-type formula plays a pivotal role. Here, we present an outline of how this Feynman--Kac-type formula is applied to the Fröhlich model, without providing a full proof. Readers interested in a more comprehensive argument are encouraged to consult the relevant literature.

\begin{Prop}
\label{fkf1}
Let \(F,G \in \mathfrak{K}_{D^{N,(p)}}\). For any \(\varepsilon>0\), we have
\[
\langle F \,\vert\, \ex^{-\beta H_\varepsilon} G\rangle
= \int_{D^{N,(p)}} \mathbb{E}_{\boldsymbol{x}}\Bigl[\ex^{S_{\rm el}} \,\mathbbm{1}_{\{\tau_p>\beta\}} \,\langle F(\boldsymbol{x}_0)\,\vert\, \varXi_\varepsilon \,G(\boldsymbol{x}_\beta)\rangle_{\F}\Bigr]\;d\boldsymbol{x},
\]
where
\begin{align*}
\varXi_\varepsilon & \deq  \ex^{ S_{\rm eff,\varepsilon}} \,\ex^{a(\vartheta_\varepsilon)^*}\,\ex^{-\beta N_{\rm ph}}\,\ex^{a(\tilde\vartheta_\varepsilon)},\\
\vartheta_\varepsilon(k) & \deq  -g_L^{1/2} \sum_{j = 1}^N \int_0^\beta \ex^{-\varepsilon k^2}\,\ex^{-\mathrm{i}k x_{j,s}}\,\ex^{-s}\,ds,\\
\tilde\vartheta_\varepsilon(k) & \deq  -g_L^{1/2} \sum_{j = 1}^N \int_0^\beta \ex^{-\varepsilon k^2}\,\ex^{\mathrm{i}k x_{j,s}}\,\ex^{-(\beta-s)}\,ds.
\end{align*}
\end{Prop}
\begin{proof}
This identity can be established through an application of the Baker-Campbell-Hausdorff formula. For detailed arguments, see, e.g., \cite[Eqs. (2.51) and (2.52)]{Hiroshima2021}.
\end{proof}

In the case \(\varepsilon=0\), one can also obtain the Feynman--Kac-type formula by formally setting \(\varepsilon=0\) in Proposition \ref{fkf1}. However, it is not immediately clear whether \(\vartheta_0,\tilde \vartheta_0 \in \ell^2(D^*)\). Thus, the mathematical validity of this substitution must be verified. This difficulty is circumvented by the following lemma:
\begin{Lemm}\label{u}
The following uniform bounds hold:
\[
\sup_{\vepsilon\ge 0}\sup_{\bs{x}\in \mathbb{R}^N}\mathbb{E}_{\bs{x}}\left[\|\vartheta_{\vepsilon}\|^2_{\ell^2}\right]<\infty,\quad \text{and}\quad
\sup_{\vepsilon\ge 0}\sup_{\bs{x}\in \mathbb{R}^N}\mathbb{E}_{\bs{x}}\left[\|\tilde{\vartheta}_{\vepsilon}\|^2_{\ell^2}\right]<\infty.
\]
In particular, it follows that \(\vartheta_0,\tilde{\vartheta}_0 \in \ell^2(D^*)\) almost surely.
\end{Lemm}

\begin{proof}
We adapt the arguments from \cite{MATTE2018}. The proofs for \(\vartheta_{\vepsilon}\) and \(\tilde \vartheta_{\vepsilon}\) are analogous, so we only present the details for \(\vartheta_0\). For each $\vepsilon\ge 0$, let
\be
\Psi_{\vepsilon}(s,\boldsymbol{x}_s) \deq -g_L^{1/2}\sum_{j=1}^N \ex^{-s}\,\ex^{-\mathrm{i}k x_{j,s}}\ex^{-\vepsilon k^2}. \label{DefPsi}
\ee
By Itō's formula,
\be
\Psi_{\vepsilon}(\beta, \bs{x}_{\beta}) + N
= -\Bigl(1 + \tfrac{N}{2}k^2\Bigr)\int_0^\beta \Psi_{\vepsilon}(s, \bs{x}_s)\,ds \;-\;\mathrm{i}\sum_j k \int_0^\beta \Psi_{\vepsilon, j}(s, x_{j, s}) \,dx_{j,s}, \label{ItoF}
\ee
where \(\Psi_{\vepsilon, j}(s, x_{j, s}) \deq -g_L^{1/2}\ex^{-s}\,\ex^{-\mathrm{i}k x_{j,s}}\ex^{-\vepsilon k^2}\). Hence,
\be
\vartheta_{\vepsilon} = \int_0^\beta \Psi_{\vepsilon}\,ds
= -\frac{\Psi_{\vepsilon}(\beta, \bs{x}_{\beta})+N}{\,1 + \tfrac{N}{2}k^2\,}
- \sum_j \frac{\mathrm{i}k}{\,1 + \tfrac{N}{2}k^2\,}\,\int_0^\beta \Psi_{\vepsilon, j} \,dx_{j,s}
\deq \mathscr{X}_{\beta, \vepsilon}+\mathscr{Y}_{\beta, \vepsilon}. \label{DefXY}
\ee
We estimate  both terms on the right-hand side separately.  We first note
\be
\| \mathscr{X}_{\beta, \vepsilon}\|^2_{\ell^2}
\le N^2(1+g_L^{1/2})^2 \sum_{k\in D^*} \left(1+\frac{N}{2}k^2\right)^{-2}, \label{BoundX}
\ee
which implies that $\sup_{\bs x} \Ex_{\bs{x}}[\|\mathscr{X}_{\beta, \vepsilon}\|^2]<\infty$. Next,
\begin{align*}
\mathbb{E}_{\bs x}\left[
\|\mathscr{Y}_{\beta, \vepsilon}\|_{\ell^2}^2
\right]
&=\sum_{i, j}
\sum_{k\in D^*}
\left|\frac{\im k}{1+\frac{N}{2}k^2}\right|^2 
\mathbb{E}_{\bs x}\left[
\int_0^\beta \Psi_{\vepsilon, i} dx_{i, s}
\int_0^\beta \Psi_{\vepsilon, j}dx_{j, s}
\right]
\\
&=\sum_{i, j}
\sum_{k\in D^*}
\left|\frac{\im k}{1+\frac{N}{2}k^2}\right|^2 
\delta_{i, j}
\left(\int_0^\beta \mathbb{E}_{\bs x}[|\Psi_{\vepsilon, i}|^2] ds\right)
\\
&\le N g_L
\sum_{k\in D^*}
\left|\frac{\im k}{1+\frac{N}{2}k^2}\right|^2 
\frac{1}{2}(1-\ex^{-2\beta })
<\infty. 
\end{align*}
This completes the proof.
\end{proof}

By Lemma \ref{u}, the operators \(a^\dagger(\vartheta_0)\) and \(a(\tilde \vartheta_0)\) are well-defined almost surely. With this in mind, the following proposition holds:

\begin{Prop}\label{FKItype}
\label{fkf2}
Let \(F,G \in \mathfrak{K}_{D^{N,(p)}}\). Then
\[
\langle F \,\vert\, \ex^{-\beta H_{\rm F}} G\rangle
= \int_{D^{N,(p)}} \mathbb{E}_{\boldsymbol{x}}\Bigl[\ex^{S_{\rm el}}\,\mathbbm{1}_{\{\tau_p>\beta\}}\,\langle F(\boldsymbol{x}_0)\,\vert\, \varXi_0\, G(\boldsymbol{x}_\beta)\rangle_{L^2(\mathcal{Q})}\Bigr]\; d\boldsymbol{x}.
\]
Here, let us recall from Lemma \ref{u} that \(\varXi_0\) is indeed well-defined.
\end{Prop}
The proof of this proposition is provided in Appendix~\ref{AppB}.

From Proposition \ref{fkf2}, we obtain
\be
\bigl(\ex^{-\beta H_{\rm F}}G\bigr)(\boldsymbol{x})
= \mathbb{E}_{\boldsymbol{x}}\Bigl[\ex^{S_{\rm el}}\,\mathbbm{1}_{\{\tau_p>\beta\}}\,\varXi_0\,G(\boldsymbol{x}_\beta)\Bigr], \label{FKtype}
\ee
for almost every \(\boldsymbol{x}\in D^{N,(p)}\).

\subsection{Proof of Theorem \ref{Hyp2} : Existence}\label{PfHyC}
It suffices to show the following:

\begin{Prop}\label{Hyp3}
Let \(t > 0\) and \(q > 2\) be such that \(t < \beta/3\) and \(2 < q < \frac{2}{1 - t/\log 4}\). Then, for every \(G \in L^2(D^{N,(p)} \times \mathcal{Q})\), there exists a constant \(C > 0\), independent of \(G\), such that
\[
\|J^{-1}\,\ex^{-\beta H_{\rm F}}\,J\,G\|_{L^q(D^{N,(p)} \times \mathcal{Q})}
\;\le\;
C\,\|G\|_{L^2(D^{N,(p)} \times \mathcal{Q})}.
\]
In particular, the operator \(J^{-1}\,\ex^{-\beta H_{\rm F}}\,J\) maps \(L^2(D^{N,(p)} \times \mathcal{Q})\)  into \(L^q(D^{N,(p)} \times \mathcal{Q})\) for all \(p \in \{1, \dots, N\}\).
\end{Prop}

To prove this proposition, we recall the following well-known fact:
\begin{Lemm}\label{hyp}
Let \(2<q<4\) and \(t > \bigl(1-\tfrac{2}{q}\bigr)\log 4\). Then
\[
\|J^{-1}\ex^{-tN_{\rm ph}}J\,F\|_{L^q(\mathcal{Q})}
\;\le\;
\|F\|_{L^2(\mathcal{Q})}.
\]
\end{Lemm}
\begin{proof}
See, e.g., \cite[Example 1.66]{Hiroshima2020}. 
\end{proof}

\subsubsection{ Proof of Proposition \ref{Hyp3}}
Note that
\[
\int_{D^{N,(p)}\times \mathcal{Q}}\!\!\Bigl\lvert J^{-1}\,\ex^{-\beta H_{\rm F}}\,J\,G(\boldsymbol{x})\Bigr\rvert^q\,d\mu\,d\boldsymbol{x}
\;=\;\int_{D^{N,(p)}} \Bigl\|J^{-1}\,\ex^{-\beta H_{\rm F}}\,J\,G(\boldsymbol{x})\Bigr\|_{L^q(\mathcal{Q})}^q\,d\boldsymbol{x}.
\]
By Lemma \ref{hyp}, \(J^{-1}\,\ex^{-tN_{\rm ph}}\,J\) is hypercontractive from \(L^2(\mathcal{Q})\) into \(L^q(\mathcal{Q})\). Therefore,
\begin{align*}
\int_{D^{N,(p)}} \Bigl\|J^{-1}\,\ex^{-\beta H_{\rm F}}\,J\,G(\boldsymbol{x})\Bigr\|_{L^q(\mathcal{Q})}^q\,d\boldsymbol{x}
&=\int_{D^{N,(p)}}
\Bigl\|J^{-1}\,\ex^{-tN_{\rm ph}}\,\ex^{tN_{\rm ph}}\,\ex^{-\beta H_{\rm F}}\,J\,G(\boldsymbol{x})\Bigr\|_{L^q(\mathcal{Q})}^q\,d\boldsymbol{x}\\
&\le\;\int_{D^{N,(p)}}
\Bigl\|J^{-1}\,\ex^{tN_{\rm ph}}\,\ex^{-\beta H_{\rm F}}\,J\,G(\boldsymbol{x})\Bigr\|_{L^2(\mathcal{Q})}^q\,d\boldsymbol{x}.
\end{align*}

From  Eq. \eqref{FKtype} we observe that
\begin{align*}
& \left\|J^{-1} \ex^{tN_{\rm ph}}\ex^{-\beta H_{\rm F}}J G(\bs{x})\right\|_{L^2(\mathcal{Q})}\\
&= \left\|\ex^{tN_{\rm ph}}\ex^{-\beta H_{\rm F}}JG(\bs{x})\right\|_{\F}\\
&= \left\|\mathbb{E}_{\bs x}\left[\mathbbm{1}_{\left\{\tau_p>\beta \right\}} \ex^{S_{\rm eff, 0}}\ex^{S_{\rm el}}\ex^{tN_{\rm ph}}\ex^{a(\vartheta_0)^*} \ex^{-\beta N_{\rm ph}} \ex^{a(\tilde \vartheta_0)} 
JG(\bs{x}_\beta)\right]\right\|_{\F}.
\end{align*}
Moreover,
\begin{align*}
&\left\|\ex^{tN_{\rm ph}}\ex^{a(\vartheta_0)^*} \ex^{-\beta N_{\rm ph}} \ex^{a(\tilde \vartheta_0)} 
JG(\bs{x}_\beta)\right\|_{\F}\\
&=\left\|\ex^{a(\ex^t\vartheta_0)^*} \ex^{tN_{\rm ph}}\ex^{-\beta N_{\rm ph}} \ex^{a(\tilde \vartheta_0)} 
JG(\bs{x}_\beta)\right\|_{\F}\\
&
=\left\|\ex^{a(\ex^t\vartheta_0)^*} \ex^{-(\beta/3)N_{\rm ph}} \ex^{tN_{\rm ph}}\ex^{-(\beta/3)N_{\rm ph}} \ex^{-(\beta/3)N_{\rm ph}}\ex^{a(\tilde \vartheta_0)} 
JG(\bs{x}_\beta)\right\|_{\F}\\
&
\leq
\left\|\ex^{a(\ex^t\vartheta_0)^*} \ex^{-(\beta/3)N_{\rm ph}}\right\| 
\cdot \left\|\ex^{tN_{\rm ph}}\ex^{-(\beta/3)N_{\rm ph}} \right\|
\cdot \left\|\ex^{-(\beta/3)N_{\rm ph}}\ex^{a(\tilde \vartheta_0)} \right\| \cdot
\|G(\bs{x}_\beta)\|_{L^2(\mathcal{Q})}\\
&
\leq
\left\|\ex^{a(\ex^t\vartheta_0)^*} \ex^{-(\beta/3)N_{\rm ph}}\right\| 
\cdot \left\|\ex^{-(\beta/3)N_{\rm ph}}\ex^{a(\tilde \vartheta_0)} \right\| 
\cdot\|G(\bs{x}_\beta)\|_{L^2(\mathcal{Q})}. 
\end{align*} 
Above, we used the identity \(\ex^{tN_{\rm ph}}\,\ex^{a(\vartheta_0)^*}=\ex^{a(\ex^t\vartheta_0)^*}\,\ex^{tN_{\rm ph}}\).

Together with Lemma \ref{bound0}, it follows that
\begin{align*}
&\int_{D^{N, (p)}}  \left\|J^{-1} \ex^{-\beta H_{\rm F}}J G(\bs{x})\right\|_{L^q(\mathcal{Q})}^q d\bs{x}\\
&\leq \int_{D^{N, (p)}}  
\left(
\mathbb{E}_{\bs x}\left[
\mathbbm{1}_{\left\{\tau_p>\beta \right\}} \ex^{ S_{ 0}}
\|\ex^{a(\ex^t\vartheta_0)^*} \ex^{-(\beta/3)N_{\rm ph}}\|  
\cdot \|\ex^{-(\beta/3)N_{\rm ph}}\ex^{a(\tilde \vartheta_0)} \| 
\cdot \|G(\bs{x}_\beta)\|\right]\right)^qd\bs{x}\\
&\leq \int_{D^{N, (p)}}  
\mathbb{E}_{\bs x}\left[\mathbbm{1}_{\left\{\tau_p>\beta \right\}}
\ex^{2 S_{0}}
\|\ex^{a(\ex^t\vartheta_0)^*} \ex^{-(\beta/3)N_{\rm ph}}\|^2  
\cdot \|\ex^{-(\beta/3)N_{\rm ph}}\ex^{a(\tilde \vartheta_0)} \|^2 \right]^{q/2}\\
&\quad \quad \times \mathbb{E}_{\bs x}\left[\mathbbm{1}_{\left\{\tau_p>\beta \right\}} \|G(\bs{x}_\beta)\|^2\right]^{q/2}d\bs{x}\\
&\leq C \int_{D^{N, (p)}}  
\mathbb{E}_{\bs x}\left[\mathbbm{1}_{\left\{\tau_p>\beta \right\}} \|G(\bs{x}_\beta)\|^2\right]^{q/2}d\bs{x}\\
&= C \int_{D^{N, (p)}}\left| \left( \ex^{-\beta T_{D^{N, (p)}}} \|G(\cdot)\|^2\right)(\bs{x})\right|^{q/2}d\bs{x}. 
 \end{align*}
Here,  $T_{D^{N,(p)}}$ denotes the Dirichlet Laplacian on $D^{N,(p)}$. 
Since \(\ex^{-\beta T_{D^{N,(p)}}}\) is hypercontractive and maps \(L^{q/2}(D^{N,(p)})\) into \(L^{1}(D^{N,(p)})\), we get
\begin{align*}
\int_{D^{N, (p)}} \left|\left(\ex^{-\beta T_{D^{N, (p)}}} \|G(\cdot)\|^2\right)(\bs{x})\right|^{q/2}d\bs{x}\leq
\left(
\int_{D^{N, (p)}}  \|G(\bs{x})\|^2 d\bs{x}\right)^{q/2}.
\end{align*}
Hence,
\[
\Bigl\|J^{-1}\,\ex^{-\beta H_{\rm F}}\,J\,G\Bigr\|_{L^q(D^{N,(p)}\times \mathcal{Q})}
\;\le\;C\,\|G\|_{L^2(D^{N,(p)}\times \mathcal{Q})}. 
\]
This completes the proof of  Proposition \ref{Hyp3}
\qed

\begin{Lemm}\label{bound0}
Let \(a,b,c,d\ge0\). Then:
\begin{align*}
\begin{array}{ll}
(1) \sup_{\vepsilon\ge 0}\sup_{\bs{x}}\mathbb{E}_{\bs x}\left[\ex^{aS_{\rm eff, \vepsilon}}\right]<\infty;& 
(2) \sup_{\bs{x}\in D^{N, (p)}}\mathbb{E}_{\bs x}\left[\mathbbm{1}_{\{\tau_p>\beta\}}\ex^{bS_{\rm el}}\right]<\infty; \\
(3) \sup_{\bs{x}}\mathbb{E}_{\bs x}\left[\left\|\ex^{a(\ex^t\vartheta_0)^*} \ex^{-(\beta/3)N_{\rm ph}}\right\|^c\right]<\infty; & 
(4) \sup_{\bs{x}}\mathbb{E}_{\bs x}\left[\left\|\ex^{-(\beta/3)N_{\rm ph}}\ex^{a(\tilde \vartheta_0)} \right\|^d\right]<\infty.
\end{array}
\end{align*}
\end{Lemm}
\begin{proof}
Statement (1) follows directly from Proposition~\ref{S-Prop}.  
Statement (2) is an immediate consequence of assumption \hyperlink{H2}{\bf(H.2)}.
By applying Proposition~\ref{bound}, statements (3) and (4) are obtained through bounds on the expectations \(\mathbb{E}_{\boldsymbol{x}}\bigl[\ex^{c\,\|\ex^t\vartheta_0\|_{\ell^2}^2}\bigr]\) and \(\mathbb{E}_{\boldsymbol{x}}\bigl[\ex^{c\,\|\tilde\vartheta_0\|_{\ell^2}^2}\bigr]\). The finiteness of these expectations is guaranteed by Proposition \ref{bound}.
\end{proof}

\subsection{Proof of Theorem \ref{Hyp2}: Uniqueness}\label{Unique}
 We demonstrate that, for every \(\beta>0\), the operator \(J^{-1} \ex^{-\beta H_{\rm F}} J\) is positivity-improving on \(L^2\bigl(D^{N,(p)};L^2(\mathcal{Q})\bigr)\).
The proof strategy is inspired by \cite[Theorem 3.51]{Hiroshima2020}.

Before proceeding with the proof, let us recall a fundamental fact. Consider the Hilbert space \(\ell^2(D^*)\) and define an involution \(c\) by
$
(cf)(k) = f^*(-k),
$
for each \(f \in \ell^2(D^*)\). A function \(f\) is called {\bf \(c\)-real} if \(cf = f\). We denote by \(\ell^2_c(D^*)\) the set of all  \(c\)-real elements.
\begin{Lemm}\label{pp}
Suppose that \(f \in \ell_c^2(D^*)\). Then both \(J^{-1} \ex^{-\beta N_{\rm ph}} \ex^{a(f)} J\) and \(J^{-1} \ex^{a(f)^*} \ex^{-\beta N_{\rm ph}} J\) preserve positivity on $L^2(\mathcal{Q})$.
\end{Lemm}

\proof
For \(f \in \ell_c^2(D^*)\), we define a multiplication operator on \(L^2(\mathcal{Q})\) by
$$
\phi(f) \;=\; 2^{-1/2} \,J^{-1}\,\overline{(\,a(f) \;+\; a(f)^*)}\,J,
$$
where the overbar on an operator denotes its closure. Let \(F \in L^2(\mathcal{Q})\) be any non-negative, non-zero function. It can be approximated by functions of the form \(F_n\bigl(\phi(f_1),\dots,\phi(f_n)\bigr)\), where \(F_n \in \mathscr{S}(\mathbb{R}^n)\) with \(F_n \ge 0\), and $f_1, \dots, f_n\in \ell^2_c(D^*)$. Here, \(\mathscr{S}(\mathbb{R}^n)\) denotes the Schwartz space on \(\mathbb{R}^n\). 
By applying the Fourier transform, we can express
$$
F_n\bigl(\phi(f_1),\dots,\phi(f_n)\bigr)
= (2\pi)^{-\tfrac{n}{2}}
\int_{\mathbb{R}^n} \hat{F}_n(k)\,\exp\left(\im \sum_{j}k_j\,\phi(f_j)\right)\, dk.
$$
From this representation, it follows that
$$
J^{-1}\ex^{a(f)}J\,F_n\bigl(\phi(f_1),\dots,\phi(f_n)\bigr)
= F_n\bigl(\phi(f_1) + \langle f | f_1 \rangle, \dots, \phi(f_n) + \langle f | f_n \rangle\bigr).
$$
Therefore,
$
J^{-1}\ex^{a(f)}J\,F_n\bigl(\phi(f_1),\dots,\phi(f_n)\bigr) \geq 0.
$
Moreover, since $J^{-1}\ex^{-\beta N_{\rm ph}}J$ is positivity-improving, we also obtain
$
J^{-1}\ex^{-\beta N_{\rm ph}}\,\ex^{a(f)}J\,F_n\bigl(\phi(f_1),\dots,\phi(f_n)\bigr) \geq 0.
$
Hence,
$
J^{-1}\ex^{-\beta N_{\rm ph}}\,\ex^{a(f)}J\,F \geq 0.
$
Taking the adjoint, we conclude that
$
J^{-1}\ex^{a(f)^*}\ex^{-\beta N_{\rm ph}}J
$
is likewise a positivity-preserving operator.
\qed

\begin{Prop}\label{PIEx}
$J^{-1}\ex^{-\beta H_{\rm F}}J$ is positivity-improving for every $\beta>0$.
\end{Prop}
\begin{proof}
Let $F,G\in  L^2(D^{N, (p)}\times \mathcal{Q})$. By Proposition \ref{fkf2}, we have
\[
\left\la F\, \Big|\,  J^{-1}\ex^{-\beta H_0} JG\right\ra=
\int_{D^{N, (p)}}{\mathbb E}_{\bs{x}}\left[\ex^{S_{\rm el}}\,\mathbbm{1}_{\{\tau_p>\beta \}}\left\la F(\bs{x}_0)\, \Big|\,  J^{-1} \varXi_0J\, G(\bs{x}_\beta )\right\ra_{L^2(\mathcal{Q})}\right]\!\,d \bs{x}.
\]
Suppose that $F, G\in L^2(D^{N,(p)};L^2(\mathcal{Q}))$ are non-negative, non-zero functions. Then $F(\bs{x})\in L^2(\mathcal{Q})$ for each $\bs{x}\in D^{N,(p)}$, and $F(\bs{x})$ is a non-negative function on $\mathcal{Q}$. Similarly, $G(\bs{x})\in L^2(\mathcal{Q})$ for each $\bs{x}\in D^{N,(p)}$, and $G(\bs{x})$ is a non-negative function on $\mathcal{Q}$. Let
$
A_F\deq \{\,\bs{x}\in D^{N,(p)}\,:\, F(\bs{x})\not\equiv0\}
$ and $
A_G \deq \{\,\bs{x}\in D^{N,(p)}\,:\, G(\bs{x})\not\equiv0\}.
$
By the assumption on $F$ and $G$, the Lebesgue measures of both $A_F$ and $A_G$ are positive. Define
$
\chi\deq \{(\bs{x}_s)_s\in C_{\bs{x}}([0,\beta];\BbbR^N)\,:\,\bs{x}_{\beta}\in A_G\}.
$
Let $P^{\beta}_{D^{N, (p)}}(\bs{x}, \bs{y})$ be the integral kernel of $\ex^{-\beta T_{D^{N,(p)}}}$. From \cite[Corollary 8.12]{grigor2012heat}, $P^{\beta}_{D^{N,(p)}}(\bs{x}, \bs{y})$ is strictly positive, so
\[
{\mathbb E}_{\bs{x}}\bigl[\mathbbm{1}_\chi\,\mathbbm{1}_{\{\tau_p>\beta \}}\bigr] 
=\int_{A_G} P^{\beta}_{D^{N,(p)}}(\bs{x}, \bs{y})\,d\bs{y} >0.
\]
By Lemma \ref{pp} and the fact that $J^{-1} \ex^{-(\beta /3)N_{\rm ph}} J$ is positivity-improving, we observe that
\begin{align*}
&\left\la F(\bs{x}_0)\,\Big|\,  J^{-1} \varXi_0J\, G(\bs{x}_\beta )\right\ra_{L^2(\mathcal{Q})}\\
&=\ex^{S_{\rm eff,0}}\left\la J^{-1} \ex^{-(\beta /3)N_{\rm ph}}\ex^{a(\bar \vartheta_0)} J\,F(\bs{x}_0)\,\Big|\,
J^{-1} \ex^{-(\beta /3)N_{\rm ph}} J\,
J^{-1}\ex^{-(\beta /3)N_{\rm ph}} \ex^{a(\tilde \vartheta_0)}J\, G(\bs{x}_\beta )\right\ra_{L^2(\mathcal{Q})}>0
\end{align*}
for every path $(\bs{x}_s)_s\in \chi$ with $\bs{x}_0\in A_F$. Therefore,
\begin{align*}
&\int_{D^{N,(p)}}{\mathbb E}_{\bs{x}}\left[\ex^{S_{\rm el}}\,\mathbbm{1}_{\{\tau_p>\beta \}}\,
\Bigl\la F(\bs{x}_0)\,\Big|\, J^{-1} \varXi_0J\, G(\bs{x}_\beta )\Bigr\ra_{L^2(\mathcal{Q})}\right]\,d \bs{x}\\
&\quad\geq
\int_{A_F}{\mathbb E}_{\bs{x}}\left[\mathbbm{1}_\chi\, \ex^{S_{\rm el}}\,\mathbbm{1}_{\{\tau_p>\beta \}}
\Bigl\la F(\bs{x}_0)\,\Big|\,  J^{-1} \varXi_0J\, G(\bs{x}_\beta )\Bigr\ra_{L^2(\mathcal{Q})}\right]\,d\bs{x}.
\end{align*}
Since the integrand
$ 
 \ex^{S_{\rm el}}\, 
\left\la F(\bs{x}_0)\,\Big|\,  J^{-1} \varXi_0J\, G(\bs{x}_\beta )\right\ra_{L^2(\mathcal{Q})}
$
is strictly positive for all paths $(\bs{x}_s)_s\in \chi$ with $\bs{x}_0\in A_F$ and the Lebesgue measure of $A_F$ is positive, it follows that
$
\bigl\langle F\,\big|\;  J^{-1}\ex^{-\beta H_{\rm F}} J\,G\bigr\rangle>0.
$
Then the proposition follows.
\end{proof}

\section{Proof of Theorem \ref{GEBasic}}\label{PfofThmGEBasic}

\subsection*{Proof of (i)}
As stated in \cite{AizenmanLieb}, we also have the following for the ground state energy of $H_{\rm el}$:
\be
E_{\rm el}(|M|)=\mathcal{E}\left(H_{\rm el} \restriction L^2(D^{N, (p)})\right). \notag
\ee
The proof of this equality is identical to that of Proposition \ref{GSEng2}.

Let \(\phi_{\rm g}\) be the normalized ground state of \(H_{\rm el} \restriction L^2(D^{N, (p)})\). By the Perron--Frobenius theorem \cite[Theorem XIII.44]{Reed1978}, \(\phi_{\rm g}\) can be chosen to be strictly positive. Hence, we obtain:
\begin{align*}
\ex^{-\beta E(|M|)}&\ge \left\la \phi_{\rm g}\otimes \vOm\,  \Big|\,  \ex^{-\beta H_{\rm F}} \phi_{\rm g}\otimes \vOm\right\ra\no
&=\int_{D^{N, (p)}} d\bs{x} \Ex_{\bs{x}}\left[\mathbbm{1}_{\left\{\tau_p>\beta \right\}}  
\phi_{\rm g}(\bs{x}_0) \phi_{\rm g}(\bs{x}_{\beta})\,  \ex^{S_0}
\right]\no
&>\int_{D^{N, (p)}} d\bs{x} \Ex_{\bs{x}}\left[\mathbbm{1}_{\left\{\tau_p>\beta \right\}}  
\phi_{\rm g}(\bs{x}_0) \phi_{\rm g}(\bs{x}_{\beta})\right]\no
&=\la \phi_{\rm g}\,  |\,  \ex^{-\beta H_{\rm el}} \phi_{\rm g}\ra=\ex^{-\beta E_{\rm el}(|M|)}.
\end{align*}
Here, the second inequality follows from Theorem \ref{PathIntF} (i). \qed

\subsection*{Proof of (ii)}

\begin{Prop}\label{EFKF}
The following holds:
\begin{align*}
E(|M|)=-\lim_{\beta\to \infty}\frac{1}{\beta}
\log \int_{D^{N, (p)}} d\bs{x}\Ex_{\bs{x}}\left[\mathbbm{1}_{\left\{\tau_p>\beta \right\}}  
\ex^{S_{0}}
\right].
\end{align*}
\end{Prop}
\begin{proof}By Theorem \ref{Hyp2},  the ground state of $H_{\rm F}$ in $\K_{D^{N, (p)}}$ is unique and can be chosen to be strictly positive.

Now, let $\omega$ be the function identically equal to $1$ on $D^{N, (p)}$. Then, $\omega \otimes \vOm$ is identically equal to $1$ on $D^{N, (p)} \times \mathcal{Q}$ and, in particular, must have non-zero overlap with the ground state of $H_{\rm F}\restriction \K_{D^{N, (p)}}$. Consequently,
\begin{align*}
\mathcal{E}\left(H_{\rm F} \restriction \K_{D^{N, (p)}}\right) = -\lim_{\beta \to \infty} \frac{1}{\beta} \log
\left\langle \omega \otimes \vOm\,  \Big|\,  \ex^{-\beta H_{\rm F}} \omega \otimes \vOm \right\rangle
\end{align*}
holds. By combining this with Proposition \ref{GSEng2} and Theorem \ref{PathIntF} (iv), we deduce the desired result.
\end{proof}

From Theorem \ref{PathIntF} (ii) and Proposition \ref{EFKF}, Theorem \ref{GEBasic} (ii) follows immediately. \qed

\section{Concluding Remarks}\label{ConRem}

Let us conclude this paper with several remarks:

\begin{itemize}
\item The results established herein remain valid for Hamiltonians with an ultraviolet cutoff, \( H_{\varepsilon} \). In this case, the intricate discussion regarding the removal of the UV cutoff in Section \ref{FKIF} is rendered unnecessary, and assumption \hyperlink{H2}{\bf (H. 2)} is also superfluous.

\item In \cite{Hakobyan2010}, Hakobyan extended the ordering of energy levels to the \(\mathrm{SU}(N)\) Hubbard model. A similar extension is feasible for the Fröhlich model analyzed in this paper. 
As an intriguing development in a direction distinct from the ordering of energy levels discussed in \cite{Hakobyan2010}, we note the work of \cite{Nachtergaele2011}.

\item In \cite{AizenmanLieb}, Aizenman and Lieb extended the ordering of energy levels to finite temperatures. Whether the results of this paper can be similarly generalized to finite temperatures remains an intriguing open problem, which is currently under investigation \cite{MiyaoEtEl}.

\item An interesting open question is how the inclusion of the quantum electromagnetic field affects the results presented here. If the interaction between the quantum electromagnetic field and the spin is neglected, the ordering of energy levels continues to hold. Conversely, we conjecture that the incorporation of such interactions may lead to a breakdown of the ordering.

\item The behavior of the ordering of energy levels in the thermodynamic limit remains unresolved even in the context of the Schrödinger operator.
\end{itemize}

\appendix

\section{Auxiliary Results on Second Quantization Operators}

In this appendix, we state without proofs several auxiliary results regarding second quantization operators, which will be utilized throughout this paper.

\begin{Prop}\label{bound}
The following operator norm inequalities hold:
\begin{align*}
\left\|\ex^{a(f)^*}\ex^{-\frac{t}{2}N_{\mathrm{ph}}}\right\|\leq 
\begin{cases}
\sqrt{2}\, \ex^{(8/s)\|f\|^2}, & 0 < s < t < 1, \\[0.5em]
\sqrt{2}\, \ex^{8\|f\|^2}, & 1 \leq t.
\end{cases}
\end{align*}
In particular, it follows that
\begin{align*}
\left\|\ex^{a(f)^*}\ex^{-tN_{\mathrm{ph}}}\ex^{a(f)}\right\|\leq 
\begin{cases}
2\, \ex^{(16/s)\|f\|^2}, & 0 < s < t < 1, \\[0.5em]
2\, \ex^{16\|f\|^2}, & 1 \leq t.
\end{cases}
\end{align*}
Moreover, the difference satisfies
\begin{align*}
\left\|\ex^{a(f)^*}\ex^{-\frac{t}{2}N_{\mathrm{ph}}}-\ex^{a(g)^*}\ex^{-\frac{t}{2}N_{\mathrm{ph}}}\right\|
\quad\leq
\begin{cases}
2\sqrt{2}\,\|f-g\|\,\ex^{(2/s)(2\|f\|+2\|g\|+1)^2}, & 0 < s < t < 1, \\[0.5em]
2\sqrt{2}\,\|f-g\|\,\ex^{2(2\|f\|+2\|g\|+1)^2}, & 1 \leq t.
\end{cases}
\end{align*}
\end{Prop}

For detailed proofs, we refer the reader, for instance, to \cite[Corollaries 1.34 and 1.35]{Hiroshima2020} and \cite{MATTE2018}.

\section{Proof of Proposition \ref{FKItype}}\label{AppB}
Proposition \ref{FKItype} can be proven by appropriately modifying and applying the ideas presented in \cite{MATTE2018} concerning the Nelson model. (For detailed explanations of the construction of the Feynman--Kac-type formula for the Nelson model, see also \cite[Eqs. (2.51) and (2.52)]{Hiroshima2021}.) Those references focus specifically on the energy renormalization needed for the Nelson model.  As we have already observed, the Fröhlich model analyzed in the present paper does not require energy renormalization upon the removal of the ultraviolet cutoff. Consequently, several aspects of the analysis from \cite{MATTE2018} become simpler in our context. Nevertheless, identifying precisely which parts may be simplified, and conducting those simplifications appropriately, turn out to be more subtle and intricate tasks than one might initially expect. In this appendix, we have carefully executed all such subtle procedures, and provide a clear outline of the resulting simplified proof.

\begin{Lemm}\label{ExTh}
For any $a \in \mathbb{R}$, it holds that
$$
\sup_{\vepsilon\ge 0}\sup_{\bs{x}\in \mathbb{R}^N} \mathbb{E}_{\bs{x}}\left[\mathrm{e}^{a \|\vartheta_{\vepsilon}\|^2_{\ell^2}}\right]<\infty,\quad\sup_{\vepsilon\ge 0}\sup_{\bs{x}\in \mathbb{R}^N} \mathbb{E}_{\bs{x}}\left[\mathrm{e}^{a \|\tilde{\vartheta}_{\vepsilon}\|^2_{\ell^2}}\right]<\infty.
$$
\end{Lemm}
\begin{proof}
We prove the statement only for $\vartheta_{\vepsilon}$, since the case of $\tilde{\vartheta}_{\vepsilon}$ can be treated analogously. Let $\mathscr{X}_{\beta, \vepsilon}$ and $\mathscr{Y}_{\beta, \vepsilon}$ be as defined in Eq. \eqref{DefXY}. Then, by the triangle inequality and the Cauchy--Schwarz inequality, we have
\be
\mathbb{E}_{\bs{x}}\left[
\mathrm{e}^{a \|\vartheta_{\vepsilon}\|^2_{\ell^2}}
\right]
 \le
\left(
\mathbb{E}_{\bs{x}}\left[
\mathrm{e}^{2 a \|\mathscr{X}_{\beta, \vepsilon}\|^2_{\ell^2}} 
\right]
\right)^{1/2}
\left(
\mathbb{E}_{\bs{x}}\left[
\mathrm{e}^{2 a \|\mathscr{Y}_{\beta, \vepsilon}\|^2_{\ell^2}} 
\right]
\right)^{1/2}.\notag
\ee
By \eqref{BoundX}, it immediately follows that $\sup_{\bs{x}}\mathbb{E}_{\bs{x}}\left[
\mathrm{e}^{2 a \|\mathscr{X}_{\beta, \vepsilon}\|^2_{\ell^2}} 
\right]<\infty$. Therefore, it suffices to show that $\sup_{\bs{x}}\mathbb{E}_{\bs{x}}\left[
\mathrm{e}^{2 a \|\mathscr{Y}_{\beta, \vepsilon}\|^2_{\ell^2}} 
\right]<\infty$. To this end, we express $\|\mathscr{Y}_{\beta, \vepsilon}\|^2_{\ell^2}$ as
$$
\|\mathscr{Y}_{\beta, \vepsilon}\|^2_{\ell^2}=\int_0^{\beta} \bs{\varUpsilon}_{\vepsilon, s} \cdot d\bs{x}_s,
$$
where $\bs{\varUpsilon}_{\vepsilon, s}=(\varUpsilon_{1, \vepsilon, s}, \dots, \varUpsilon_{N, \vepsilon,  s})$ with
$$
\varUpsilon_{j, \vepsilon, s}\deq \frac{2\im k}{1+\frac{N}{2}k^2}\left\la \mathscr{Y}_{s, \vepsilon}\ |\ \Psi_{j}(s, x_{j, s})\right\ra_{\ell^2}.
$$
Recall that the functions $\Psi_j$ are defined in the proof of Lemma \ref{u}, immediately below Eq.  \eqref{ItoF}. By an argument analogous to that in the proof of Lemma \ref{YLemm}, which employs Girsanov\rq{}s theorem and the Cauchy--Schwarz inequality, we derive the following bound:
$$
\left(
\mathbb{E}_{\bs{x}}\left[
\mathrm{e}^{a \|\mathscr{Y}_{\beta, \vepsilon}\|_{\ell^2}^2}
\right]
\right)^2
\le
\mathbb{E}_{\bs{x}}\left[
\mathrm{e}^{2a^2\int_0^{\beta} |\bs{\varUpsilon}_{\vepsilon, s} |^2 ds}
\right].
$$
We proceed to estimatee $\int_0^{\beta} |\bs{\varUpsilon}_{\vepsilon, s} |^2 ds$ from above. Let $\vartheta_{\vepsilon, s}$ denote the function obtained by replacing $\beta$ with $s$ in the definition of $\vartheta_{\vepsilon}$. Applying Itō's formula \eqref{ItoF}, we obtain
\begin{align*}
\int_0^{\beta} |\bs{\varUpsilon}_{\vepsilon, s} |^2 ds=&4\sum_{j=1}^N \int_0^{\beta}\left|
\left\la \vartheta_{\vepsilon, s} - \mathscr{X}_{\vepsilon, s} \ \Bigg|\ \frac{\im k }{1+\frac{N}{2}k^2} \Psi_{\vepsilon, j}(s, x_{j, s})
\right\ra_{\ell^2} 
\right|^2ds\no
\le & 8\sum_{j=1}^N \int_0^{\beta}\left|
\left\la \vartheta_{\vepsilon, s} \ \Bigg|\ \frac{\im k }{1+\frac{N}{2}k^2} \Psi_{\vepsilon, j}(s, x_{j, s})
\right\ra_{\ell^2} 
\right|^2ds\no
&\quad +8\sum_{j=1}^N \int_0^{\beta}\left|
\left\la \mathscr{X}_{\vepsilon, s}\ \Bigg|\ \frac{\im k }{1+\frac{N}{2}k^2} \Psi_{\vepsilon, j}(s, x_{j, s})
\right\ra_{\ell^2} 
\right|^2ds\no
\deq & I_1+I_2.
\end{align*}
If we define $\Psi_{\vepsilon}$ as in Eq. \eqref{DefPsi}, then $I_1$ admits the following estimate:
\begin{align}
I_1 & \le 16\sum_{j=1}^N \int_0^{\beta}\left|
\left\la \int_0^s \Psi_{\vepsilon}(x, \bs{x}_t)dt \ \Bigg|\ \frac{\im k }{1+\frac{N}{2}k^2} \Psi_{\vepsilon, j}(s, x_{j, s})
\right\ra_{\ell^2} 
\right|^2ds\no
&\le 
16\sum_{k\in D^*}
\left(
\frac{k}{1+\frac{N}{2}k^2}
\right)^2 \int_0^{\beta}g_L^2 N^3(1-\mathrm{e}^{-s})^2 \mathrm{e}^{-2s} ds<\infty. \notag
\end{align}
Therefore, $\sup_{\vepsilon\ge 0}\sup_{\bs{x}} \sup_{(\bs{x}_s)_s} I_1<\infty$.
As for $I_2$, using \eqref{BoundX}, we immediately conclude that $\sup_{\vepsilon\ge 0}\sup_{\bs{x}} \sup_{(\bs{x}_s)_s} I_2<\infty$.
\end{proof}

\begin{Lemm}\label{Th^8}
The following limits hold:
$$
\displaystyle \lim_{\vepsilon\to +0} \sup_{\bs{x}\in \mathbb{R}^N}\mathbb{E}_{\bs{x}}\left[
\|\vartheta_{\vepsilon}-\vartheta_0\|^8_{\ell^2}
\right]=0,
\quad
\displaystyle \lim_{\vepsilon\to +0} \sup_{\bs{x}\in \mathbb{R}^N}\mathbb{E}_{\bs{x}}\left[
\|\tilde{\vartheta}_{\vepsilon}-\tilde{\vartheta}_0\|^8_{\ell^2}
\right]=0.
$$
\end{Lemm}

\begin{proof}
Using \eqref{DefXY}, we obtain
$$
\|\vartheta_{\vepsilon}-\vartheta_0\|^8_{\ell^2} \le 2^7 \left(
\|\mathscr{X}_{\beta, \vepsilon}-\mathscr{X}_{\beta, 0}\|_{\ell^2}^8
+
\|\mathscr{Y}_{\beta, \vepsilon}-\mathscr{Y}_{\beta, 0}\|_{\ell^2}^8
\right).
$$
We estimate each term on the right-hand side separately. First, noting that
$$
\|\mathscr{X}_{\beta, \vepsilon}-\mathscr{X}_{\beta, 0}\|_{\ell^2}^2 \le \sum_{k\in D^*}\frac{(1-\ex^{-\vepsilon k^2})^2}{(1+\frac{N}{2}k^2)^2},
$$
we immediately conclude that
$
\lim_{\vepsilon\to +0} \sup_{\bs{x}} \|\mathscr{X}_{\beta, \vepsilon}-\mathscr{X}_{\beta, 0}\|_{\ell^2}^8 = 0.
$

On the other hand, applying Jensen\rq{}s inequality $\mathbb{E}_{\bs{x}}[Z^{1/4}] \ge (\mathbb{E}_{\bs{x}}[Z])^{1/4}$ and the Cauchy--Schwarz inequality yields
$$
\mathbb{E}_{\bs{x}} \left[
\|\mathscr{Y}_{\beta, \vepsilon}-\mathscr{Y}_{\beta, 0}\|_{\ell^2}^8
\right]
\le \left\{
\sum_{k\in D^*}\left(
\mathbb{E}_{\bs{x}} \left[|\mathscr{Y}_{\beta, \vepsilon}-\mathscr{Y}_{\beta, 0}|^8\right]
\right)^{1/4}
\right\}^4.
$$
Furthermore, using  the Burkholder--Davis--Gundy inequality (see, for example, \cite{Lrinczi2020}), we estimate
\begin{align*}
\mathbb{E}_{\bs{x}} \left[|\mathscr{Y}_{\beta, \vepsilon}-\mathscr{Y}_{\beta, 0}|^8\right]
\le &\sum_{j=1}^N 2^7 \mathbb{E}_{\bs{x}} \left[
\left|
\int_0^{\beta} \frac{\im k}{1+\frac{N}{2}k^2}(\Psi_{\vepsilon, j}-\Psi_{0, j})\, dx_{j, s}
\right|^8
\right] \\
\le &2^{15} \cdot 7^4 \cdot \beta^4 \cdot g_L^4\cdot N \left|
\frac{k (1 - \ex^{-\vepsilon k^2})}{1+\frac{N}{2}k^2}
\right|^8.
\end{align*}
Combining this estimate with the previous inequality, we thus obtain that
\[
\lim_{\vepsilon\to +0} \sup_{\bs{x} \in \mathbb{R}^N} \mathbb{E}_{\bs{x}} \left[
\|\mathscr{Y}_{\beta, \vepsilon}-\mathscr{Y}_{\beta, 0}\|_{\ell^2}^8
\right] = 0,
\]
as required.
\end{proof}

For each $\varepsilon \ge 0$, we define
$$
\mathscr{P}_{\vepsilon}(\bs{x})\deq \Ex_{\bs{x}}\left[
\ex^{S_{\rm el}} \mathbbm{1}_{\{\tau_p> \beta\}} \left\la
F(\bs{x}_0)\ |\ \varXi_{\vepsilon} G(\bs{x}_{\beta})
\right\ra_{\F}
\right].
$$

\begin{Prop}\label{PropP}
The following statements hold:
\begin{itemize}
\item[\rm (i)] For all $\vepsilon \ge 0$, it holds that $\mathscr{P}_{\vepsilon}\in L^1(D^{N,(p)})$.
\item[\rm (ii)] $\mathscr{P}_{\vepsilon}$ converges to $\mathscr{P}_0$ in $L^1(D^{N,(p)})$ as $\vepsilon\to +0$.
\end{itemize}
\end{Prop}

\begin{proof}
For each $\vepsilon \ge 0$, define the operator $Q_{\vepsilon}$ by
\[
Q_{\vepsilon}\deq \ex^{a(\vartheta_{\vepsilon})^*}\ex^{-\beta N_{\rm ph}}\ex^{a(\tilde{\vartheta}_{\vepsilon})}.
\]
Note here that Lemma \ref{u} ensures that the right-hand side is almost surely well-defined even when $\vepsilon=0$. Proposition \ref{bound} gives the inequality
$
\|Q_{\vepsilon}\|\le 2\ex^{16(\|\vartheta_{\vepsilon}\|_{\ell^2}^2+\|\tilde{\vartheta}_{\vepsilon}\|_{\ell^2}^2)}.
$
Therefore, by applying the Cauchy--Schwarz inequality and Lemma \ref{ExTh}, we obtain
\be
\sup_{\vepsilon\ge 0}\sup_{\bs{x}\in \mathbb{R}^N} \Ex_{\bs{x}} [\|Q_{\vepsilon}\|^4]<\infty. \label{Q^4}
\ee
Consequently,
\begin{align*}
&\int_{D^{N,(p)}}|\mathscr{P}_{\vepsilon}(\bs{x})|\,d\bs{x}\\
&\quad \le   \ex^{-\beta \inf U}\int_{D^{N,(p)}} \|F(\bs{x})\|_{\mathscr{F}}
\left(\Ex_{\bs{x}}\left[\ex^{2S_{\mathrm{eff},\vepsilon}}\|Q_{\vepsilon}\|^2\right]\right)^{1/2}
\left(\Ex_{\bs{x}}\left[\mathbbm{1}_{\{\tau_p>\beta\}}\|G(\bs{x}_{\beta})\|_{\mathscr{F}}^2\right]\right)^{1/2}d\bs{x}\\
&\quad \le   C\|F\|\|G\|,
\end{align*}
where
\[
C=\ex^{-\beta \inf U}\sup_{\vepsilon\ge 0}\sup_{\bs{x}\in\mathbb{R}^N}\left(\Ex_{\bs{x}}\left[\ex^{4S_{\mathrm{eff},\vepsilon}}\right]\Ex_{\bs{x}}\left[\|Q_{\vepsilon}\|^4\right]\right)^{1/4}.
\]
By Lemma \ref{bound0} and inequality \eqref{Q^4}, it follows that $C$ is finite. Thus, the proof of (i) is completed.

Now, we prove statement (ii). By the Cauchy--Schwarz inequality,
\begin{align*}
&\|\mathscr{P}_{\vepsilon}-\mathscr{P}_0\|_{L^1}\\
&\quad\le \ex^{-\beta \inf U}\int_{D^{N,(p)}}\|F(\bs{x})\|_{\mathscr{F}}
\left(\Ex_{\bs{x}}\left[\|\ex^{S_{\mathrm{eff},\vepsilon}}Q_{\vepsilon}-\ex^{S_{\mathrm{eff},0}}Q_0\|^2\right]\right)^{1/2}
\left(\Ex_{\bs{x}}\left[\mathbbm{1}_{\{\tau_p>\beta\}}\|G(\bs{x}_{\beta})\|_{\mathscr{F}}^2\right]\right)^{1/2}d\bs{x}.
\end{align*}
We next bound the factor involving the difference between $Q_{\vepsilon}$ and $Q_0$ appearing in the inequality above. By the triangle inequality, we have
\begin{align*}
\Ex_{\bs{x}}\left[\|\ex^{S_{\mathrm{eff},\vepsilon}}Q_{\vepsilon}-\ex^{S_{\mathrm{eff},0}}Q_0\|^2\right]
&\le 2\Ex_{\bs{x}}\left[\|\ex^{S_{\mathrm{eff},\vepsilon}}Q_{\vepsilon}-\ex^{S_{\mathrm{eff},0}}Q_{\vepsilon}\|^2\right]+2\Ex_{\bs{x}}\left[\|\ex^{S_{\mathrm{eff},0}}Q_{\vepsilon}-\ex^{S_{\mathrm{eff},0}}Q_0\|^2\right]\\
&\deq I_1+I_2.
\end{align*}
We analyze the terms $I_1$ and $I_2$ individually. For the first term $I_1$, we  use the inequality
\[
I_1/2\le \left(\Ex_{\bs{x}}\left[|\ex^{S_{\mathrm{eff},\vepsilon}}-\ex^{S_{\mathrm{eff},0}}|^4\right]\right)^{1/2}\left(\Ex_{\bs{x}}[\|Q_{\vepsilon}\|^4]\right)^{1/2}.
\]
Using \eqref{Q^4} and \eqref{DConv0}, we deduce
\[
\sup_{\vepsilon\ge 0}\sup_{\bs{x}\in\mathbb{R}^N}|I_1|<\infty\quad\text{and}\quad\lim_{\vepsilon\to +0}I_1=0.
\]
Next, we estimate the second term $I_2$. By applying the Cauchy--Schwarz inequality again, we obtain
\[
I_2/2\le \left(\Ex_{\bs{x}}\left[\|Q_{\vepsilon}-Q_0\|^4\right]\right)^{1/2}\left(\Ex_{\bs{x}}\left[\ex^{4S_{\mathrm{eff},0}}\right]\right)^{1/2}.
\]
Further bounding the first factor on the right-hand side above, Proposition \ref{bound} yields
\[
\|Q_{\vepsilon}-Q_0\|\le \|\vartheta_{\vepsilon}-\vartheta_0\|_{\ell^2}J_{1,\vepsilon}+\|\tilde{\vartheta}_{\vepsilon}-\tilde{\vartheta}_0\|_{\ell^2}J_{2,\vepsilon},
\]
where we define
\begin{align*}
J_{1,\vepsilon}&\deq 4\ex^{4(2\|\vartheta_{\vepsilon}\|_{\ell^2}+2\|\vartheta_0\|_{\ell^2}+1)^2}\ex^{16\|\tilde{\vartheta}_{\vepsilon}\|_{\ell^2}^2},\\
J_{2,\vepsilon}&\deq 4\ex^{4(2\|\tilde{\vartheta}_{\vepsilon}\|_{\ell^2}+2\|\tilde{\vartheta}_0\|_{\ell^2}+1)^2}\ex^{16\|\vartheta_0\|_{\ell^2}^2}.
\end{align*}
Thus, using the Cauchy--Schwarz inequality again, we have
\begin{align*}
\Ex_{\bs{x}}\left[\|Q_{\vepsilon}-Q_0\|^4\right]&\le 8\left(\Ex_{\bs{x}}\left[\|\vartheta_{\vepsilon}-\vartheta_0\|_{\ell^2}^8\right]\right)^{1/2}\left(\Ex_{\bs{x}}[J_{1,\vepsilon}^8]\right)^{1/2}\\
&\quad+8\left(\Ex_{\bs{x}}\left[\|\tilde{\vartheta}_{\vepsilon}-\tilde{\vartheta}_0\|_{\ell^2}^8\right]\right)^{1/2}\left(\Ex_{\bs{x}}[J_{2,\vepsilon}^8]\right)^{1/2}.
\end{align*}
Applying Lemmas \ref{ExTh} and \ref{Th^8}, we find
$
\lim_{\vepsilon\to +0}\sup_{\bs{x}\in\mathbb{R}^N}I_2=0.
$
Finally, by combining these analyses with Lebesgue's dominated convergence theorem, we conclude that statement (ii) holds.
\end{proof}

\subsubsection*{Completion of the Proof of Proposition \ref{FKItype}}
By Proposition \ref{PropP}, we obtain
\[
\langle F\,  |\,  \ex^{-\beta H_{\rm F}} G \rangle = \lim_{\vepsilon \to +0}\langle F\, |\, \ex^{-\beta H_{\vepsilon}} G\rangle 
= \lim_{\vepsilon \to +0}\int_{D^{N, (p)}} \mathscr{P}_{\vepsilon}(\bs{x})\,d\bs{x}
= \int_{D^{N, (p)}} \mathscr{P}_{0}(\bs{x})\,d\bs{x}.
\]
This completes the proof.  
\qed

\bibliographystyle{abbrvurl}

\end{document}